\documentclass[11pt]{elsarticle}

\usepackage{amsmath}
\usepackage{amsthm}
\usepackage{fullpage}
\usepackage{stmaryrd}
\usepackage{amsfonts}
\usepackage{amssymb}
\usepackage{float}
\usepackage{url}
\usepackage{hyperref}
\usepackage{tikz}
  \usetikzlibrary{automata}
  \usetikzlibrary{arrows}
  \usetikzlibrary{shapes}
  \usetikzlibrary{decorations.pathmorphing}
  \usetikzlibrary{fit}
  \usetikzlibrary{matrix}
  \tikzstyle{every picture}=[
    >=stealth', shorten >=1pt, node distance=1.44cm,auto,bend angle=45,initial text=,
    every state/.style={inner sep=0.75mm, minimum size=1mm},font=\scriptsize,
  ]
\usepackage{algorithm}
\usepackage{algorithmic}

\newsavebox\dotbox
\sbox{\dotbox}{\(\displaystyle\bigodot\)}
\DeclareMathOperator*{\bigcdot}{\raisebox{0pt}[\ht\dotbox][\dp\dotbox]{\(\bullet\)}}
\newcommand*{\DashedArrow}[1][]{\mathbin{\tikz [baseline=-0.25ex,-latex, dashed,#1] \draw [#1] (0pt,0.5ex) -- (1.3em,0.5ex);}}%
\newcommand*{\Arrow}[1][]{\mathbin{\tikz [baseline=-0.25ex,-latex, #1] \draw [#1] (0pt,0.5ex) -- (1.3em,0.5ex);}}%

\newtheorem{definition}{Definition}
\newtheorem{theorem}{Theorem}
\newtheorem{lemma}{Lemma}
\newtheorem{proposition}{Proposition}
\newtheorem{corollary}{Corollary}
\newdefinition{example}{Example}
\newdefinition{question}{Question}

\usepackage{etoolbox}
\newcommand\modif[2]{%
   {#2}%
}

\begin{document} 

\begin{frontmatter}

  \title{(k,l)-Unambiguity and Quasi-Deterministic Structures}
  
  \author[rouen]{Pascal Caron}
     \ead{pascal.caron@univ-rouen.fr}
  \author[lehavre]{Marianne Flouret}
    \ead{marianne.flouret@univ-lehavre.fr}
  \author[depart]{Ludovic Mignot} 
     \ead{ludovic.mignot@univ-rouen.fr}

  \address[rouen]{LITIS, Universit\'e de Rouen, 76801 Saint-\'Etienne du Rouvray Cedex, France}

  \address[lehavre]{LITIS, Universit\'e du Havre, 76058 Le Havre Cedex, France}
    \address[depart]{D\'epartement d'informatique, Universit\'e de Rouen, 76801 Saint-\'Etienne du Rouvray Cedex, France}

  \begin{abstract} 
    We focus on the family of $(k,l)$-unambi\-guous automata that encompasses the one of deterministic $k$-lookahead automata introduced by Han and Wood.
    We show that this family presents nice theoretical properties that allow us to compute quasi-deterministic structures. 
    These structures are smaller than DFAs and can be used to solve the membership problem faster than NFAs.    
  \end{abstract}
  
\end{frontmatter}


\section{Introduction}\label{se:int}
  \modif{This paper is an extended version of~\cite{CFM14}.}{}

One of the \modif{best known}{most popular} automata construction is the position automaton construction \cite{Glu61}. If a regular expression has $n$ occurrences of symbols, then the corresponding position automaton, which is not necessarily deterministic, has exactly $n+1$ states.
  The $1$-unambiguous  regular languages have been defined by Br\"uggemann-Klein and Wood~\cite{BW98} as  languages denoted  by  regular expressions  the position automata of which are deterministic. They have also shown that there exist regular languages that are not $1$-unambiguous. This property has practical implication, since it models a property needed in XML DTDs~\cite{BPS06}. Indeed, XML DTDs are defined as an extension of classical context-free grammars in which the right hand side of any production is a one-unambiguous regular expression. 
   Consequently, \modif{}{a} characterization of such languages, \modif{that}{which} has been considered \emph{via} the deterministic minimal automaton, is very important, since it proves that not all the regular languages can be used in XML DTDs.  
  The computation of a small deterministic recognizer is also technically important since it allows a reduction of the time and of the space needed to solve the membership problem (to determine whether or not a given word belongs to a language).
   As a consequence, one may wonder whether there exists a family of languages encompassing the $1$-unambiguous one that can be recognized by a polynomial-size deterministic family of recognizers.
  
  On the one hand, numerous extensions of $1$-unambiguity have been considered, like $k$-block determinism~\cite{GMW01}, $k$-lookahead determinism~\cite{HW08} or weak $1$-unambiguity~\cite{CHM11}. All of these extensions, likely to the notion of $1$-unambiguity, are expression-based properties. A regular language is $1$-unambiguous (resp. $k$-block deterministic, $k$-lookahead deterministic, weakly $1$-unambi\-guous) if it is denoted by a $1$-unambiguous (resp. $k$-block deterministic, $k$-lookahead deterministic, weakly $1$-unambi\-guous) regular expression. All of these three properties are defined through a recognizer construction.
  
  On the other hand, the concept of lookahead delegation, introduced in \cite{DIS04}, handles determinism without computing a deterministic recognizer; the determinism is simulated by a fixed number of input symbols read ahead, in order to select the \emph{right} transition in the NFA. This concept arose in a formal study of web-services composition and its practical applications~\cite{GHIS04}. Questions about complexity and decidability of lookahead delegation have been answered by Ravikumar and Santean in~\cite{RS07}. Finally, having defined predictable semiautomata, Brzozowski and Santean~\cite{BS09} improved complexity of determining whether an automaton admits a lookahead delegator.
  
  The notion of $(k,l)$-unambiguity for automata is the first step of the study of the $(k,l)$-unambiguity for languages. In this paper,  
   we define the notion of $(k,l)$-unambiguity for automata, leading to the computation of quasi-deterministic structures, that are smaller than DFAs and that can be used to solve the membership problem faster than NFAs.    
  These structures act as automata for which a window of size $k$ and some shifting states are added. Recognizing a word on such a structure is performed as follows: At the beginning of the process, the window matches the $k$ first letters of the input word. When a shifting state is reached and the input word is not entirely read, the window is slided along the input word ($j<k$ letters, depending on the shifting state), the \modif{QDS}{Quasi-Deterministic Structures (QDS)} returns in a regular state and the reading restarts at the beginning of the window.
  We then show, thanks to an equivalence relation, how to reduce such structures. We also exhibit a family of languages for which reduced QDS are exponentially smaller than minimal DFAs.
  Next step is to study the $(k,l)$-unambiguous languages, that are languages denoted by some  regular expressions the position automaton of which is $(k,l)$-unambiguous. Having such a regular expression allows us to directly compute a quasi-deterministic structure to solve the membership problem.
  
  In Section~\ref{se:klna}, after defining the $(k,l)$-unambiguity as an extension of $k$-lookahead determinism, we characterize this notion making use of the \emph{square automaton}. 
  In Section~\ref{se:qds}, we define quasi-deterministic structures that allow us to perform a constant space membership test. Section~\ref{se:klnanfa2qds} is devoted to the computation of the quasi-deterministic structure associated with a $(k,l)$-unambiguous automaton. The notion of quotient of a quasi-deterministic structure is defined in Section~\ref{sec qds}, and a right invariant equivalence relation is investigated. It is shown in Section~\ref{sec qds vs dfa} that reduced quasi-deterministic structures can be exponentially smaller than minimal \modif{deteministic}{deterministic} automata.
  
  \modif{}{This paper is an extended version of~\cite{CFM14}.}

\section{Preliminaries}\label{se:pre}

Let $\varepsilon$ be \emph{the empty word}. An \emph{alphabet} $\Sigma$ is a finite set of distinct symbols. The usual concatenation of symbols is denoted by $\cdot$, and $\varepsilon$ is its identity element. We denote by $\Sigma^*$ the smallest set containing $\Sigma\cup\{\varepsilon\}$ and closed under  the $\cdot$ operation. Any subset of $\Sigma^*$ is called a \emph{language over} $\Sigma$. Any element of $\Sigma^*$ is called a \emph{word}. The length of a word $w$, noted $|w|$, is the number of symbols in $\Sigma$ \modif{it is the concatenation of}{occurring in $w$}  (\emph{e.g.} $|\varepsilon|=0$). By extension the number of elements of a set $S$ is denoted by $|S|$.
For a given integer $k$, we denote by $\Sigma^k$ the set of words of length $k$ and by $\Sigma^{\leq k}$ the set $\bigcup_{k'\leq k} \Sigma^{k'}$. 
 Let $w=a_1\cdots a_{|w|}$ be a word in $\Sigma^*$ such that for any $k$ in $[1,|w|]$, $a_k$ is a symbol in $\Sigma$. 
Let $i$ and $j$ be two integers such that $i\leq j\leq |w|$. We denote by $w[i,j]$ the subword $a_i\cdots a_j$ of $w$ starting at position $i$ and ending at the position $j$ and by $w[i]$ the $i$-th symbol $a_i$ of $w$.
More generally, we will define by $\bigcdot_{j=i}^ka_j$ the word $a_i\cdots a_k$.  In case $i>k$, this word is $\varepsilon$.
%
%
%
%

A \emph{nondeterministic finite automaton} (\textbf{NFA}) $A$ is a $5$-tuple $(\Sigma,Q,I,F,\delta)$ where $\Sigma$ is an alphabet, $Q$ is a \emph{set of states}, $I\subset Q$ is a \emph{set of initial states}, $F\subset Q$ is a \emph{set of final states} and $\delta$ is a \emph{transition function} defined from $Q\times\Sigma$ to $2^{Q}$. The function $\delta$ can be interpreted as a subset of $Q\times\Sigma\times Q$ defined by $q'\in\delta(q,a)$ $\Leftrightarrow$ $(q,a,q')\in\delta$.
 The domain of $\delta$ is extended to $2^Q\times\Sigma^*$ as follows: for any symbol $a$ in $\Sigma$, for any state $q$ in $Q$, for any subset $P$ of $Q$, for any word $w$ in $\Sigma^*$: $\delta(P,\varepsilon)=P$, $\delta(P,a)=\bigcup_{p\in P}\delta(p,a)$, $\delta(P,aw)=\delta(\delta(P,a),w)$. 
  Let  $k$ be an integer and $w$ be a word in $\Sigma^k$. A \emph{path} $p$ \emph{labelled by} $w$ is a finite sequence $t=(p_0,\ldots, p_k)$\modif{}{ of states} such that for any integer $0\leq j<k$, $p_{j+1}\in \delta(p_j,w[j+1])$. The path $t$ \emph{starts} with $p_0$. Two paths $t=(p_0,\ldots, p_k)$ and $t'=(p'_0,\ldots, p'_k)$ labelled by $w$ are \emph{totally distinct} if for any integer $0< j\leq k$, $p_j\neq p'_j$.
  A path $t=(p_0,\ldots, p_k)$ is a cycle if $k>0$ and $p_0=p_k$.  
The automaton $A$ is  \emph{deterministic} if the two following properties hold: $|I|=1$ and $\forall (q,a)\in Q\times \Sigma,\ |\delta(q,a)|\leq 1$.
A state $q$ in $Q$ is  \emph{accessible} (resp. \emph{coaccessible}) if there exists a word $w$ in $\Sigma^*$ such that $q\in\delta(I,w)$ (resp. $\delta(q,w)\cap F\neq\emptyset$). The automaton $A$ is \emph{accessible} (resp. \emph{coaccessible})  if any state in $Q$ is accessible (resp. coaccessible). The automaton is \emph{trim} if any state in $Q$ is accessible and coaccessible.
  
Given a word $w$ and an $n$-state automaton $A$, the membership test~\cite{HMU07}, \emph{i.e.} deciding whether $w$ belongs to $L(A)$, can be performed in time $O(n^2\times |w|)$ and in space $O(n)$. Let us suppose that $A'$ is the  $n'$-state deterministic automaton of $A$ (computed as the classical accessible part of the powerset automaton of $A$). The membership test can be performed in time $O(|w|)$ and in space $O(1)$, but $n'$ can be exponentially greater than $n$.

Glushkov~\cite{Glu61} and McNaughton and Yamada~\cite{MY60} have independently defined the construction of the \emph{Glushkov automaton} or \emph{position automaton} $G_E$ of a 
regular expression $E$. The number of states $s$ of $G_E$ is a linear function of the width \modif{$n$}{$|E|$} of $E$ (\emph{i.e.} the number of occurrences of the symbols of $\Sigma$ in $E$); in fact, \modif{$s=n+1$}{$s=|E|+1$}.
The automaton $G_E$ is a 
$(|E|+1)$-state  automaton that recognizes $L(E)$.

A regular expression $E$ is \emph{deterministic} if and only if its Glushkov automaton is. A language is  $1$-\emph{unambiguous} if there exists a deterministic expression to denote it. Br\"uggemann-Klein and Wood~\cite{BW98} have shown that determining whether a regular language is $1$-unambiguous or not is a decidable problem. Furthermore, they proposed a characterization and showed that both $1$-unambiguous languages and non $1$-unambiguous regular languages exist.

The notion of $k$-lookahead determinism~\cite{HW08} extends the one of $1$-unambi\-guity of expressions. In that purpose, Han and Wood define the $k$\emph{-lookahead deterministic position automaton} of an expression.

\begin{definition}[\cite{HW08}]\label{def kla}
  Let $A=(\Sigma,Q,I,F,\delta)$ be a position automaton of an expression. Then $A$ is a \emph{deterministic} $k$\emph{-lookahead automaton} if for any state $q_0$ in $Q$, where $(q_0,a_0,q_0)$, $(q_0,a_1,q_1)$, $\ldots$, $(q_0,a_m,q_m)$ are the out-transitions of $q_0$, with $q_i\neq q_j$ for $0\leq i,j\leq m$, it holds: $a_i\cdot \mathbb{F}_{k-1}(q_i) \cap a_j\cdot \mathbb{F}_{k-1}(q_j)=\emptyset$, where $0\leq i< j\leq m$ and $\mathbb{F}_{k-1}(q_i)$ is the set of words of length $k-1$ that labels a path starting at $q_i$.
\end{definition}

Notice that this definition can be extended to any automaton that is not a position one.
Informally, an automaton is $k$-lookahead deterministic if and only if for any state $q$, for any word $w=a_1\cdots a_k$ of length $k$, all  paths from $q$ labelled by $w$ have a common first transition  (see Figure~\ref{fig ex general klh}).
\modif{}{An automaton is lookahead-deterministic if there exists an integer $k$ such that it is $k$-lookahead-deterministic.}

\begin{figure}[H]
  \begin{minipage}{0.49\linewidth}
    \centerline{ 
      \begin{tikzpicture}[node distance=1.5cm,bend angle=30]   
	    \node[state] (q) {$q$};
	    \node (q') [right of=q]{};
	    \node[state] (q1) [above right of=q']{$q_1$};
	    \node[state] (q2) [right of=q']{$q_2$};
	    \node[state] (q3) [below right of=q']{$q_3$};
	    \path[->,dashed]
	      (q)   edge node {$w$} (q1)
	      (q)   edge node {$w$} (q2)
	      (q)   edge node {$w$} (q3);	            
      \end{tikzpicture}
    }   
    \end{minipage}
    \hfill    
  \begin{minipage}{0.49\linewidth}
    \centerline{  
      \begin{tikzpicture}[node distance=2.2cm,bend angle=25,transform shape,scale=0.75]   
	    \node[state] (q) {$q$};  
	    \node[state] (q'') [node distance=1.7cm, below right of=q] {$q''$};
	    \node (q''') [right of=q''] {$\emptyset$};
	    \node[state] (q') [above right of=q]{$q'$};
	    \node[state] (q1) [above right of=q']{$q_1$};
	    \node[state] (q2) [right of=q']{$q_2$};
	    \node[state] (q3) [below right of=q']{$q_3$};
	    \path[->]
	      (q)   edge[swap]  node {$a_1$} (q'')
	      (q'')   edge[dashed]  node {$a_2\cdots a_k$} (q''')
	      (q)   edge node {$a_1$} (q')
	      (q')   edge[dashed] node {$a_2\cdots a_k$} (q1)
	      (q')   edge[dashed] node {$a_2\cdots a_k$} (q2)
	      (q')   edge[dashed] node {$a_2\cdots a_k$} (q3);	            
      \end{tikzpicture}     
    }   
    \end{minipage}
    \caption{The $k$-lookahead determinism}
    \label{fig ex general klh}
  \end{figure}
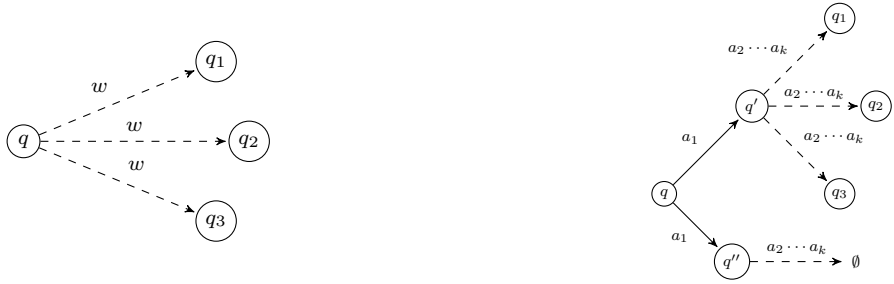
  
  Brzozowski and Santean~\cite{BS09} introduced the notion of predictability for an automaton and linked it to the one of lookahead determinism: as far as an automaton admits a unique initial state, it is $k$-predictable if and only if it is $(k+1)$-lookahead deterministic.
  
  In order to decide whether a given automaton is predictable, they make use of the square automaton defined as follows: let $A=(\Sigma,Q,I,F,\delta)$. The \emph{square automaton} $s_A$ of $A$ is the automaton $(\Sigma,Q\times Q, I\times I, F\times F, \delta')$ where for any pair $(q_1,q_2)$ of states in $Q$, for any symbol $a$ in $\Sigma$, $\delta'((q_1,q_2),a)=\delta(q_1,a)\times \delta(q_2,a)$.
  
  Finally, \modif{from the square automaton,}{} they define the pair automaton\modif{of}{, the subautomaton of the square automaton restricted to the} critical subsets of $Q$ (the set of initial states and the sets of successors of a \modif{fork}{state with at least two distinct successors}). An automaton is predictable if and only if its pair automaton admit no cycle. A closely related method has already been applied in comparable settings for Moore machines~\cite{Koh90}.

\section{The (k,l)-unambiguity}\label{se:klna}


  The definition of $k$-lookahead determinism can be extended by the introduction of an additional parameter $l$. The maximal length of ambiguity in two distinct paths from the same state and labelled by a same word is bounded by this parameter. Hence\modif{}{,} an automaton is said to be $(k,l)$-unambiguous ($l\leq k$) if and only if for any state $q$, for any word $w=a_1\cdots a_k$ of length $k$, if there exist at least two distinct paths from $q$ labelled  by $w$,  then there exists an integer $i\leq l$ such that all these paths share a common successor after a path of length $i$ (see Figure~\ref{fig ex general klna}).

\begin{figure}[H]
  \begin{minipage}{0.49\linewidth}
    \centerline{ 
      \begin{tikzpicture}[node distance=1.5cm,bend angle=30]   
	    \node[state] (q) {$q$};
	    \node (q') [right of=q]{};
	    \node[state] (q1) [above right of=q']{$q_1$};
	    \node[state] (q2) [right of=q']{$q_2$};
	    \node[state] (q3) [below right of=q']{$q_3$};
	    \path[->, dashed]
	      (q)   edge node {$w$} (q1)
	      (q)   edge node {$w$} (q2)
	      (q)   edge node {$w$} (q3);	            
      \end{tikzpicture}
    }   
    \end{minipage}
    \hfill    
  \begin{minipage}{0.49\linewidth}
    \centerline{  
      \begin{tikzpicture}[node distance=2.2cm,bend angle=30,scale=0.8]   
	    \node[state] (q) {$q$};  
	    \node[state] (q'') [node distance=1.8cm, below right of=q] {$q''$};
	    \node (q''') [node distance=2cm, right of=q''] {$\emptyset$};
	    \node[state] (q') [above right of=q]{$q'$};
	    \node[state] (q1) [above right of=q']{$q_1$};
	    \node[state] (q2) [right of=q']{$q_2$};
	    \node[state] (q3) [below right of=q']{$q_3$};
	    \path[->,dashed]
	      (q)   edge [bend left] node {$a_1\cdots a_i$} (q')
	      (q)   edge[swap]  node {$a_1\cdots a_i$} (q'')
	      (q'')   edge [swap] node {$a_{i+1}\cdots a_k$} (q''')
	      (q)   edge node {} (q')
	      (q)   edge[bend right] node {} (q')
	      (q')   edge node {$a_{i+1}\cdots a_k$} (q1)
	      (q')   edge node {$a_{i+1}\cdots a_k$} (q2)
	      (q')   edge node {$a_{i+1}\cdots a_k$} (q3);	            
      \end{tikzpicture}     
    }   
    \end{minipage}
    \caption{The $(k,l)$-unambiguity}
    \label{fig ex general klna}
  \end{figure}
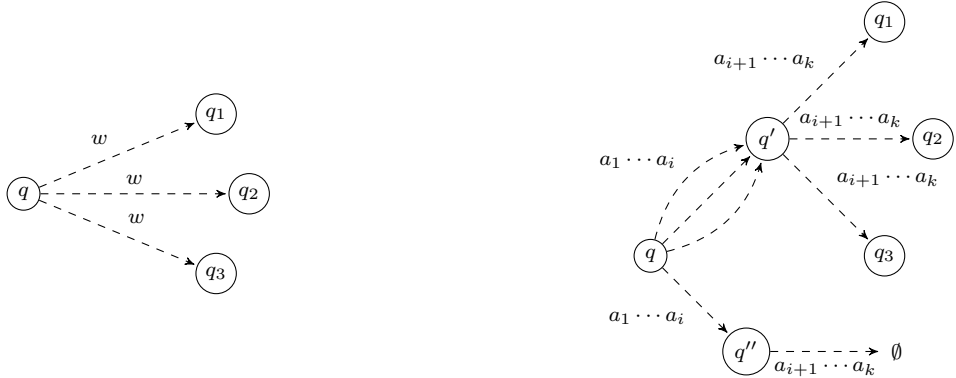 
  
\begin{definition}\label{def aut klna}
  Let $k$ and $l$ be two integers such that $l\leq k$. A finite automaton $A=(\Sigma,Q,I,F,\delta)$ is $(k,l)$-\emph{unambiguous} if $|I|=1$ and if for any state $q$ in $Q$, for any word $w$ in $\Sigma^k$, there exists an integer $1\leq i\leq l$ such that:
  
  \centerline{
    $|
      \{q'\in Q : q'\in\delta(q,w[1,i])\ \wedge\ \delta(q',w[i+1,k])\neq\emptyset\}
    |\leq 1$.
  }
\end{definition}

As a direct consequence of this definition, it holds that any $(k,l)$-unambi\-guous automaton is also a $(k,l+1)$-unambiguous automaton whenever $l<k$.

The following example enlightens the notion of $(k,l)$-unambiguity while illustrating the difference between $(k,l)$-unambiguity and $k$-lookahead determinism.

\begin{example}\label{ex aut klna pour cas concrets}
  Let us consider the automaton $A=(\Sigma,Q,I,F,\delta)$ in Figure~\ref{fig ex aut klna pour cas concrets}.
  Let us notice that for $q=q_0$, $w=aba$, for all $1\leq i\leq 3$, $|\delta(q_0,w[1,i])|> 1$. 
  As a consequence, the automaton is not $(3,i)$-unambiguous.
  Increasing the length $k$ of the window allows us to avoid this ambiguity. Indeed, for any word $w$ of length $4$, $|\delta(q_0,w)|\leq 1$. Hence $A$ is $(4,4)$-unambiguous.
  Furthermore, $A$ is also $(4,3)$-unambiguous 
  but not $(4,2)$-unambiguous.
  Finally, let us notice that this automaton is not $k$-lookahead deterministic for any integer $k$ since for any integer $j$ and for any prefix $w=aw'$ of $(abaa)^j$, $\delta(q_0,a)=\{1,2\}$ and $w'\in \mathbb{F}_{|w'|}(1)\cap \mathbb{F}_{|w'|}(2)$. 
\end{example}

\begin{figure}[H]
  \centerline{ 
    \begin{tikzpicture}[node distance=1.5cm,bend angle=30]   
	    \node[initial,state] (0) {$q_0$};
	    \node[state] (1) [above right of=0]{$q_1$};
	    \node[state] (2) [below right of=0]{$q_2$};
	    \node[state] (3) [right of=1]{$q_3$};
	    \node[state] (4) [right of=2]{$q_4$};
	    \node[state] (5) [right of=3]{$q_5$};
	    \node[state] (6) [right of=4]{$q_6$};
	    \node[state,accepting] (7) [right of=5]{$q_7$};
	    \node[state,accepting] (8) [right of=6]{$q_8$};
	    \path[->]
	      (0)   edge node {$a$} (1)
	      (0)   edge[swap] node {$a$} (2)
	      (1)   edge node {$b$} (3)
	      (2)   edge[swap] node {$b$} (4)
	      (3)   edge node {$a$} (5)
	      (3)   edge node {$a$} (6)
	      (4)   edge[swap] node {$a$} (6)
	      (5)   edge node {$b$} (7)
	      (6)   edge node {$c$} (8)
	      (6)   edge[swap] node {$a$} (0);	            
    \end{tikzpicture}
  }   
  \label{fig ex aut klna pour cas concrets}
  \caption{The automaton of Example~\ref{ex aut klna pour cas concrets}.}
  \end{figure}
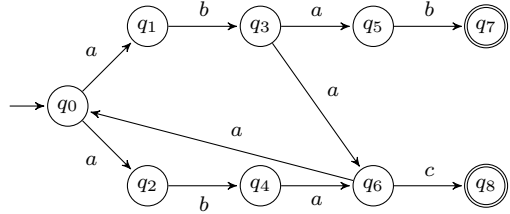

\modif{}{Let us now explicit the difference between the $k$-lookahead determinism and the $(k,l)$-unambiguity. First, as a direct consequence of Definition~\ref{def kla} and Definition~\ref{def aut klna}, the following proposition holds.}
\begin{proposition}\label{prop eq k1na klh}
  An automaton is deterministic $k$-lookahead if and only if it is  $(k,1)$-unambiguous.
\end{proposition}
\modif{\begin{proof}
  Let $A=(\Sigma,Q,I,F,\delta)$ be an automaton, $k$ be an integer and $q$ be a state in $Q$.
  \begin{enumerate}
    \item Suppose that $A$ is not $(k,1)$-unambiguous. Hence there exists a word $w$ of length $k$ such that for any integer $i\leq k$:
  
  \centerline{
    $\mathrm{Card}(
      \{q'\in Q\mid q'\in\delta(q,w[1,i])\ \wedge\ \delta(q',w[i+1,k])\neq\emptyset\}
    )> 1$.
  }
  
  Then \modif{exists}{exist} two distinct states $q_1$ and $q_2$ in $Q$ such that $q_1\in\delta(q,w[1,1])$ $\wedge$ $\delta(q_1,w[2,k])\neq\emptyset$ and $q_2\in\delta(q,w[1,1])\ \wedge\ \delta(q_2,w[2,k])\neq\emptyset$.
  Either $k=1$ and $A$ is not deterministic (\emph{i.e.} $A$ is not deterministic $1$-lookahead), or $k>1$ and by definition of $\mathbb{F}_{k-1}(q_1)$ and $\mathbb{F}_{k-1}(q_2)$, $w[2,k]\in \mathbb{F}_{k-1}(q_1)\cap \mathbb{F}_{k-1}(q_2)$. Hence 
  
  \centerline{
    $w[1,1]\cdot \mathbb{F}_{k-1}(q_1) \cap w[1,1]\cdot \mathbb{F}_{k-1}(q_2)\neq \emptyset$.
  }
  
  As a consequence, $A$ is not deterministic $k$-lookahead.
  
  \item Suppose that $A$ is not deterministic $k$-lookahead. Then there \modif{exists}{exist} three states $q_0$, $q_1$ and $q_2$ such that the transitions $(q_0,a,q_1)$ and $(q_0,a,q_2)$ belongs to $\delta$ and $a\cdot \mathbb{F}_{k-1}(q_1) \cap a\cdot \mathbb{F}_{k-1}(q_2)\neq \emptyset$. Hence there exists a word $w'$ of length $k-1$ such that $\delta(q_1,w')\neq\emptyset$ and $\delta(q_2,w')\neq\emptyset$. Consequently the states $q_1$ and $q_2$ both belong to the set $\{q'\in Q\mid q'\in\delta(q,a)\ \wedge\ \delta(q',w')\neq\emptyset\}$ and then $A$ is not $(k,1)$-unambiguous.
  \end{enumerate}
  
\end{proof}}{}

\begin{proposition}
  \modif{T}{For any integer $k$, t}here exists  a \modif{$(k,l)$}{$(k,k-1)$}-unambiguous automaton which is  not \modif{$k'$-}{}lookahead deterministic\modif{ for any integer $k'$}{}. 
\end{proposition}
\begin{proof}
  An illustration is given in Example~\ref{ex aut klna pour cas concrets}.
  \modif{}{This example can be easily generalized by considering the $(k,k-1)$-unambiguous automaton $A_k$ in Figure~\ref{general fig ex aut kl}.}
\end{proof}

\modif{}{
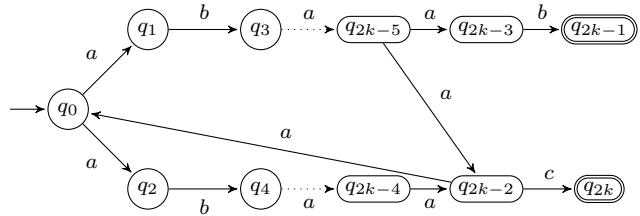
\begin{figure}[H]
  \centerline{ 
    \begin{tikzpicture}[node distance=1.5cm,bend angle=30]   
	    \node[initial,state] (0) {$q_0$};
	    \node[state] (1) [above right of=0]{$q_1$};
	    \node[state] (2) [below right of=0]{$q_2$};
	    \node[state] (2p) [right of=2]{$q_4$};
	    \node[state] (3) [right of=1]{$q_3$};
	    \node[state,rounded rectangle] (3p) [right of=3]{$q_{2k-5}$};
	    \node[state,rounded rectangle] (4) [right of=2p]{$q_{2k-4}$};
	    \node[state,rounded rectangle] (5) [right of=3p]{$q_{2k-3}$};
	    \node[state,rounded rectangle] (6) [right of=4]{$q_{2k-2}$};
	    \node[state,accepting,rounded rectangle] (7) [right of=5]{$q_{2k-1}$};
	    \node[state,accepting,rounded rectangle] (8) [right of=6]{$q_{2k}$};
	    \path[->]
	      (0)   edge node {$a$} (1)
	      (0)   edge[swap] node {$a$} (2)
	      (1)   edge node {$b$} (3)
	      (2)   edge[swap] node {$b$} (2p)
	      (2p)   edge[swap,dotted] node {$a$} (4)
	      (3)   edge[dotted] node {$a$} (3p)
	      (3p)   edge node {$a$} (5)
	      (3p)   edge node {$a$} (6)
	      (4)   edge[swap] node {$a$} (6)
	      (5)   edge node {$b$} (7)
	      (6)   edge node {$c$} (8)
	      (6)   edge[swap] node {$a$} (0);	            
    \end{tikzpicture}
  }   
  \caption{A $(k,k-1)$-unambiguous automaton.}  \label{general fig ex aut kl}

\end{figure}
}

Whenever an automaton is not $(k,l)$-unambiguous for any couple $(k,l)$ of integers \modif{}{(\emph{e.g.} when two distinct states with a loop can be reached from the same state by the same word)}, there exists a state from which it cannot be decided without ambiguity which successor will appear during the run. Hence there exists an infinite hesitation between two paths, that can be decided \emph{via} the square automaton.

\begin{theorem}\label{theorem caract klna}
  Let $A=(\Sigma,Q,I,F,\delta)$ be an accessible automaton and $P$ be the accessible part of its square-automaton. The two following propositions are equivalent:
  \begin{enumerate}
    \item there exists a couple $(k,l)$ such that $A$ is $(k,l)$-unambiguous,
    \item every cycle in $P$ contains a pair $(p,p)$ for some $p$ in $Q$.  \end{enumerate}
\end{theorem}

In order to prove Theorem~\ref{theorem caract klna}, let us first state the following lemmas.

\begin{lemma}\label{lem 2chem imp 2 chem tot dis}
  Let $A=(\Sigma,Q,I,F,\delta)$ be an automaton, $w\in \Sigma^*$ and $q\in Q$. The two following conditions are equivalent:
  \begin{itemize}
  \item  for any positive integer $k<|w|$, $|\{q': q'\in \delta(q,w[1,k]) \wedge \delta(q',w[k+1,|w|])\neq\emptyset\}
  |\geq 2$,
 \item  there \modif{exists}{exist} at least two totally distinct paths labelled by $w$ that starts with $q$.
\end{itemize}
\end{lemma}
\begin{proof}
  Let $k\geq 1$ be an integer. Let $w$ be a word in $\Sigma^k$ and $q$ be a state in $Q$ 
  such that for any integer $1\leq j\leq |w|$, $|\{q'\in \delta(q,w[1,j]) : \delta(q',w[j+1,|w|])\neq\emptyset\}|\geq 2$. Then there \modif{exists}{exist} two paths $t=(p_0,\ldots, p_k)$ and $t'=(p'_0,\ldots, p'_k)$ labelled by $w$ that starts with $q$. If these two paths are totally distinct, then the lemma is valid, otherwise there exists a third path $t''=(p''_0,\ldots, p''_k)$ such that for any integer $0\leq j\leq k$, $p_j=p'_j$ $\Rightarrow$ $p''_j\neq p_j$. Let us show by recurrence on $j$ that there \modif{exists}{exist} two totally distinct paths $r=(r_0,\ldots, r_j)$ and $r'=(r'_0,\ldots, r'_j)$ from $q$ labelled by $w[1,j]$ such that $\{r_j,r'_j\}\in\{p_j,p'_j,p''_j\}$. 
  \textbf{(a)} Let us set $j=1$. Then by definition of $t''$, either $p_1\neq p'_1$ or $p_1\neq p''_1$. \modif{R}{The r}ecurrence hypothesis is satisfied.
  \textbf{(b)} Let us set $j > 1$. Let us suppose that there \modif{exists}{exist} two totally distinct paths labelled by $w[1,j]$. Without loss of generality, let us suppose that $r_j=p_j$ and $r'_j=p'_j$. If $p_{j+1}\neq p'_{j+1}$, adding these two distinct states respectively to the paths $r$ and $r'$ constructs two totally distinct paths labelled by $w[1,j+1]$, otherwise, and without loss of generality, let us suppose that $p'_j=p''_j$. Since by definition of $t''$, $p_{j+1}=p'_{j+1}$ $\Rightarrow$ $p''_{j+1}\neq p_{j+1}$. Considering $r_{j+1}=p_{j+1}$ and $r'_{j+1}=p''_{j+1}$,  two totally distinct paths labelled by $w[1,j+1]$ are computed.

\end{proof}

\modif{}{The  following lemma is straightforward and is useful  for  proving  Theorem~\ref{theorem caract klna}.}

\modif{\begin{lemma}\label{lem 2chem total dis 1 chem car}
  Let $A$ be an automaton and $P$ be its square-automaton. Let $w$ be a word in $\Sigma^*$ and $q$ be a state in $Q$. If there \modif{exists}{exist} two totally distinct paths labelled by $w$ that starts with $q$ in $A$, then there exists a path $p=(p_0,\ldots,p_k)$ in $P$ labelled by $w$ starting with $(q,q)$ such that for any integer $1\leq j\leq k$, $p_j=(c,c')$ with $c\neq c'$.
\end{lemma}}{}
\modif{\begin{proof}
  Let $A=(\Sigma,Q,I,F,\delta)$ and $P=(\Sigma,Q',I',F',\delta')$.
  
  Let $p_1=({p_1}_j)_{0\leq j\leq k}$ and $=({p_2}_j)_{0\leq j\leq k}$ be two totally distinct paths labelled by $w$ such that ${p_1}_0={p_2}_0=q$. Let $p=(({p_1}_j,{p_2}_j))_{0\leq j\leq k}$. Let us show by recurrence on the length of $w$ that $p$ is a path in $P$ labelled by $w$ that starts with $(q,q)$ such that for any integer $1\leq j\leq k$, $p_j=(c,c')$ with $c\neq c'$.
  
  Since ${p_1}_0={p_2}_0=q$, it holds that $p_0=(q,q)$.
  
  Let $k=1$. Since $p_1$ and $p_2$ are two distinct path labelled by $w$, then $({p_1}_1,{p_2}_1)\in\delta(q,w[1])\times \delta(q,w[1])$. Hence, $({p_1}_1,{p_2}_1)\in\delta'((q,q),w[1])$.  Furthermore, since $p_1$ and $p_2$ are totally distinct, ${p_1}_1\neq {p_2}_1$.
  
  Let $k\geq 2$. Suppose that for any integer $k'< k$, $p=(p_j)_{0\leq j\leq k'}$ in $P$ is labelled by $w[1,k']$ starting with $(q,q)$ such that for any integer $1\leq j\leq k'$, $p_j=({p_1}_j,{p_2}_j)$ with ${p_1}_j\neq {p_2}_j$. Since by definition of $p_1$ and $p_2$, ${p_1}_k$ belongs to $\delta({p_1}_{k-1},w[k])$ and ${p_2}_k$ belongs to $\delta({p_2}_{k-1},w[k])$, it holds that $({p_1}_k,{p_2}_k)$ belongs to $\delta'(({p_1}_{k-1},{p_2}_{k-1}),w[k])$. Hence $({p_1}_k,{p_2}_k) \in \delta'(({p_1}_{k-1},{p_2}_{k-1}),w)$. Furthermore, since $p_1$ and $p_2$ are totally distinct, ${p_1}_k\neq {p_2}_k$.
  
\end{proof}}{}

\begin{lemma}\label{lem q acc ds a qq acc ds car a}
  Let $A=(\Sigma,Q,I,F,\delta)$ be an automaton and $P=(\Sigma,Q',I',F',\delta')$ be its square-automaton. Let $w$ be a word in $\Sigma^*$ and $q_1$ and $q_2$ be two states in $Q$. If $q_2\in\delta(q_1,w)$ then $(q_2,q_2)\in\delta'((q_1,q_1),w)$.
\end{lemma}
\modif{\begin{proof}
  By recurrence on the length of $w$.
  
  If $|w|=0$, since $q_1\in\delta(q_1,\varepsilon)$ and $(q_1,q_1)\in\delta'((q_1,q_1),\varepsilon)$, the recurrence property holds.
  
  Let $|w|> 0$. Let us suppose that $q_2\in\delta(q_1,w[1,|w|-1]])$ and $(q_2,q_2)\in\delta'((q_1,q_1),w[1,|w|-1])$. Let $q'\in\delta(q_2,w[k])$. Hence $(q',q')\in\delta'((q_2,q_2),w[k])$. As a consequence, $(q',q')\in\delta'((q_1,q_1),w)$.
  
\end{proof}}{}

\begin{proof}[\modif{}{Proof of }Theorem~\ref{theorem caract klna}]
  Let us set $A=(\Sigma,Q,\{0\},F,\delta)$ and $P=(\Sigma,Q',I',F',\delta')$.
  
  $\mathbf{(\neg 2 \Rightarrow \neg 1)}$ Let us suppose that there exists a cycle $C$ in $P$ that does not contain any pair $(p,p)$ for all state $p$ in $Q$. As a consequence, there exists a path $T$ from $(0,0)$ to a state $s=(c,c')$ in $C$ such that any predecessor of the first occurrence of $s$ does not belong to $C$. Let $q$ be the state in $Q$ such that \textbf{(a)} $(q,q)$ appears on the path $T$ from $(0,0)$ to the first occurrence of $(c,c')$ and \textbf{(b)} there exists no state $p$ in $Q$ such that $(p,p)$ appears on the path $T$ between $(q,q)$ and the first occurrence of $(c,c')$. Notice that $q$ exists since $0$ satisfies the previous propositions. Hence for any integer $k\geq 1$, there exists a word $w$ in $\Sigma^k$ such that $\delta'((q,q),w)\neq\emptyset$ and such that $|\delta(q,w)|\geq 2$. Consequently, there exists no couple $(k,l)$ such that $A$ is $(k,l)$-unambiguous.
  
  $\mathbf{(\neg 1 \Rightarrow \neg 2)}$ Let us suppose that for every integer $k$, there exists a word $w$ in $\Sigma^k$ and a state $q$ in $Q$ such that 
  for any integer $i\leq k$, $|\{q' : q'\in \delta(q,w[1,i]) \wedge \delta(q',w[i+1,k])\neq\emptyset\}|\geq 2$. Hence according to Lemma~\ref{lem 2chem imp 2 chem tot dis}, 
  there \modif{exists}{exist} at least two totally distinct paths labelled by $w$ that start with $q$. Since $q$ is reachable from $0$, then it holds from Lemma~\ref{lem q acc ds a qq acc ds car a} that $(q,q)$ belongs to $Q'$ since it is reachable from $(0,0)$. According to 
 the definition of distinct paths, for any integer $k$, there exists a word in $\Sigma^k$ such that there exists a path $(p_0,\ldots, p_k)$ in $P$ labelled by $w$ starting with $(q,q)$ such that for any integer $1\leq j\leq k$, $p_j=(c,c')$ with $c\neq c'$. Finally, whenever $k\geq |Q|\times (|Q|-1)$, there \modif{exists}{exist} two integers $1\leq k_1<k_2\leq k$ such that $p_{k_1}=p_{k_2}$. Consequently there exists a cycle in $P$ that contains no pair $(p,p)$ for any $p$ in $Q$.  
  
  \end{proof}

Notice that Theorem~\ref{theorem caract klna} defines a polynomial decision procedure to test if, for a given NFA $A$, there exists a couple $(k,l)$ of integer\modif{}{s} such that $A$ is $(k,l)$-unambiguous.

\modif{}{In order to have an upper bound of the complexity of this decision procedure, let us consider a pair automaton $P$  of $n^2$ states. It is sufficient to remove all the states $(p,p)$ of $P$ and to check if the obtained automaton is acyclic, which can be done by applying $n^2$ times the linear time Tarjan  algorithm \cite{Tar72}  which leads to a complexity in $o(n^4)$.}

The next section is devoted to the definition of quasi-deterministic structures. These structures allow us to solve the membership problem with the same complexity as deterministic automata 
while being possibly exponentially smaller. Finally, we show in Section~\ref{se:klnanfa2qds} how to convert a $(k,l)$-unambiguous NFA into a quasi-deterministic structure.

\section{The quasi-deterministic structure}\label{se:qds}

A quasi-deterministic structure is 
a structure 
derived from an automaton: it embeds a second transition function
 that is 
used to shift the input window 
(of a fixed length) 
while reading a word (see Figure~\ref{fig ex qds graph}). 
In the following, the symbol $\bot$ is used to represent undefined states and transitions.

\begin{definition}
  A \emph{quasi-deterministic structure} (QDS) is a $8$-tuple $S=(\Sigma, m, \Gamma , {\cal Q},0,$ $F,\delta,\gamma)$ where:
  \begin{itemize}
    \item $\Sigma$ is the alphabet of words,
    \item $m$ is the number of levels,
    \item $\Gamma\subset [1,m-1]$ is the alphabet of shifts,
    \item ${\cal Q}= \bigcup _{j=1}^mQ_j$ is a family of $m$ disjoint set\modif{}{s} of states (levels),
    \item $0\in Q_1$ is the initial state,
    \item $F\subset {\cal Q}$ is the set of final states,
    \item $\delta$ is a \modif{}{total} function from $Q_j\times \Sigma$ to $Q_{j+1}\cup\{\bot\}$ for $j\in\modif{\{1,\ldots,m-1\}}{[1,m-1]}$,
    \item $\gamma$ is a \modif{}{total} function from $Q_m$ to $\Gamma  \times Q_1$.  \end{itemize}
  
  The function $\delta$ can be extended for any state $q$ in $\cal Q$,\modif{ for any state $q'$ in $Q_m$,}{} for any word $w$ in $\Sigma^*$ and for any symbol $a$ in $\Sigma$ to $\delta(q,\varepsilon)=q$, \modif{$\delta(q',a)=\bot$}{$q\in Q_m \Rightarrow \delta(q,a)=\bot$}, $\delta(\bot, a)=\bot$, $\delta(q,aw)=\delta(\delta(q,a),w)$. We will denote $\tilde{\gamma}$ (resp. $\overline{\gamma}$) the restriction of the function $\gamma$ to $\Gamma$ (resp. $Q_1$). The functions $\delta$ and $\gamma$ can also be seen as sets of triplets. 
  An \emph{edge} is an element of $\delta\cup \gamma$. Two edges $(p,x,p')$ and  $(q,y,q')$ are consecutive if $p'=q$.
  A \emph{path} in a QDS is a sequence $((q_1,x_1,q_2),(q_2,x_2,q_3),\ldots,$ $(q_{n-1},x_{n-1},q_n),(q_n,x_n,q_{n+1}))$ of \modif{}{consecutive} edges.  

\end{definition}

 
 An example of a QDS is given by Figure \ref{fig ex qds graph}. 
 
\begin{figure}[H]
\begin{minipage}{0.6\linewidth}
\begin{itemize}
\item $\Sigma=\{a,b\}$, $m=3$, $\Gamma=\{1,2\}$
\item ${\cal Q}= (Q_1,Q_2,Q_3)=(\{q_1,q_6\},\{q_2,q_3,q_7\},\{q_4,q_5,q_8\})$
\item $0=q_1$
\item $F=\{q_2,q_7\}$
\item $\begin{array}{llr}\delta=\{&(q_1,a,q_2), (q_1,b,q_3),(q_2,b,q_4),\\
& (q_2,a,q_5), (q_3,a,q_5),(q_3,b,q_5),\\
&(q_6,a,q_7),(q_6,b,q_7),(q_7,a,q_8) &\}
\end{array}$
\item $\gamma=\{(q_5,2,q_1),(q_4,1,q_6),(q_8,2,q_6)\}$
\end{itemize}
\end{minipage}
\begin{minipage}{0.39\linewidth}
    \modif{\begin{tikzpicture}[node distance=2cm,bend angle=30]   
	    \node[initial, state] (1) {$1$};     
	    \node[accepting, state, right of=1] (2) {$2$}; 
	    \node[state, right of=2] (4) {$4$}; 
	    \node[state, above of=2] (3) {$3$};  
	    \node[state, above of=4] (5) {$5$};  
	    \node[state, below of=1] (6) {$6$};  
	    \node[accepting, state, right of=6] (7) {$7$}; 
	    \node[state, right of=7] (8) {$8$};  
	    \path[->]
	      (1)   edge  node {$a$} (2)
	      (1)   edge  node {$b$} (3)
	      (2)   edge  node {$a$} (5)
	      (2)   edge  node {$b$} (4)
	      (3)   edge  node {$a,b$} (5)
	      (6)   edge  node {$a,b$} (7)
	      (7)   edge  node {$a$} (8)
	      (5)   edge[dotted,swap, out=135, in = 90]  node {$2$} (1)
	      (4)   edge[dotted]  node {$1$} (6)
	      (8)   edge[dotted,bend left]  node {$2$} (6)
	      ;	            
      \end{tikzpicture}  }{
    \begin{tikzpicture}[node distance=2cm,bend angle=30]   
	    \node[initial, state] (1) {$q_1$};     
	    \node[accepting, state, right of=1] (2) {$q_2$}; 
	    \node[state, right of=2] (4) {$q_4$}; 
	    \node[state, above of=2] (3) {$q_3$};  
	    \node[state, above of=4] (5) {$q_5$};  
	    \node[state, below of=1] (6) {$q_6$};  
	    \node[accepting, state, right of=6] (7) {$q_7$}; 
	    \node[state, right of=7] (8) {$q_8$};  
	    \path[->]
	      (1)   edge  node {$a$} (2)
	      (1)   edge  node {$b$} (3)
	      (2)   edge  node {$a$} (5)
	      (2)   edge  node {$b$} (4)
	      (3)   edge  node {$a,b$} (5)
	      (6)   edge  node {$a,b$} (7)
	      (7)   edge  node {$a$} (8)
	      (5)   edge[dotted,swap, out=135, in = 90]  node {$2$} (1)
	      (4)   edge[dotted]  node {$1$} (6)
	      (8)   edge[dotted,bend left]  node {$2$} (6)
	      ;	            
      \end{tikzpicture}  }
  \end{minipage}
  \caption{The quasi-deterministic structure $S=(\Sigma, m, \Gamma,{\cal Q},0,F,\delta,\gamma)$.}
  \label{fig ex qds graph}
\end{figure}
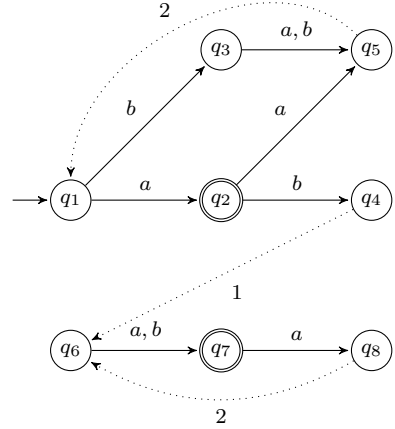

\modif{Such a structure}{In a classical automaton, a  path is successful if it starts from an initial state and ends on a final one. A QDS} can also be used as a recognizer. \modif{: a word $w$ is recognized if it labels a successful path}{} However, the label of a path in a QDS has a different meaning.
\modif{Indeed, some factors of the word can be repeated all along the path.}{Indeed, a word is read in a window of size $m-1$ (where $m$ is the number of levels) which is shifted at each $\gamma$-transition. So, a factor of this word can be read several times.}  

We define the extended transition function in a QDS. This new definition allows us to define the language recognized by a QDS.
\begin{definition}\label{def ext delta}
  Let $S=(\Sigma, m, \Gamma, {\cal Q},0,F,\delta,\gamma)$ be a QDS.  The \emph{extended transition function of} $S$ is the function $\Delta$ from $Q_1 \times \Sigma^*$ to ${\cal Q} \cup\{\bot\}$ defined for any pair $(q,w)$ in $Q_1\times\Sigma^*$ by:
  
  \centerline{
    $\Delta(q,w)=
      \left\{
        \begin{array}{l@{\ }l}
          \delta(q,w) & \text{ if } |w|\leq m-1,\\
                    \bot &\text{ if } |w|\geq m\ \wedge  \delta(q,w[1,m-1])=\bot\\ 
          \Delta(q',w[j+1,|w|])\\
          \ \ \ \  \text{ where } \gamma(\delta(q,w[1,m-1]))=(q',j) & \text{ otherwise.}
        \end{array}
      \right.$
  }
\end{definition}

\begin{definition}\label{def lang qds}
   Let $S=(\Sigma, m, \Gamma, {\cal Q},0,F,\delta,\gamma)$ be a quasi-deterministic structure. The \emph{language of} $S$ is the language $L(S)$ defined by:
   
   \centerline{
     $L(S)=\{w\in\Sigma^*\mid \Delta(0,w)\in F\}$.
   }
\end{definition}

\begin{example}\label{ex membership prob qds}
  Let us consider the structure $S$ defined in Figure~\ref{fig ex qds graph}. Let $w=bbbaabab$. \modif{F}{The f}ollowing computation illustrates that \modif{$\Delta(1,w)=7$}{$\Delta(q_1,w)=q_7$}, and since $\modif{7}{q_7}\in F$, it holds that $w\in L(S)$.
  
  
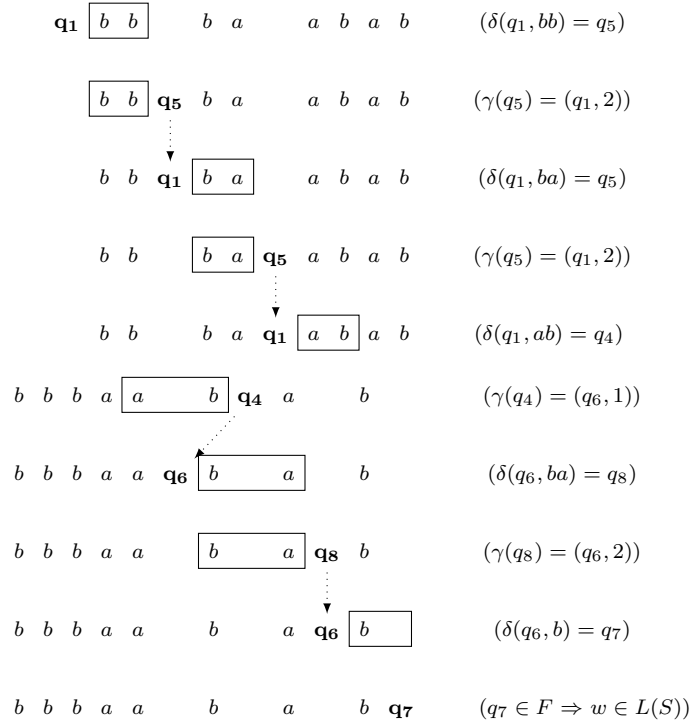
\centerline{
    \begin{tikzpicture}
      \matrix (mat) [matrix of nodes,ampersand replacement=\&,row sep=0.5cm]{
          $\modif{\mathbf{1}}{\mathbf{q_1}}$ \& $b$ \& $b$ \&     \& $b$ \& $a$ \&     \& $a$ \&     \& $b$ \&     \& $a$ \&     \& $b$ \&\& \ \ \ \ \ \ ($\delta(\modif{1}{q_1},bb)=\modif{5}{q_5}$)\\
           \& $b$ \& $b$ \& $\modif{\mathbf{5}}{\mathbf{q_5}}$ \& $b$ \& $a$ \&     \& $a$ \&     \& $b$ \&     \& $a$ \&     \& $b$\& \& \ \ \ \  \ \ ($\gamma(\modif{5}{q_5})=(\modif{1}{q_1},2)$)\\
                     \& $b$ \& $b$ \& $\mathbf{\modif{1}{q_1}}$ \& $b$ \& $a$ \&     \& $a$ \&     \& $b$ \&     \& $a$ \&     \& $b$ \&\& \ \ \ \ \ \ ($\delta(\modif{1}{q_1},ba)=\modif{5}{q_5}$)\\
                                 \& $b$ \& $b$ \&     \& $b$ \& $a$ \& $\mathbf{\modif{5}{q_5}}$ \& $a$ \&     \& $b$ \&     \& $a$ \&     \& $b$ \&\& \ \ \ \ \ \ ($\gamma(\modif{5}{q_5})=(\modif{1}{q_1},2)$)\\
            \& $b$ \& $b$ \&     \& $b$ \& $a$ \& $\mathbf{\modif{1}{q_1}}$ \& $a$ \&     \& $b$ \&     \& $a$ \&     \& $b$ \&\& \ \ \ \ \ ($\delta(\modif{1}{q_1},ab)=\modif{4}{q_4}$)\\
      };
      \node [inner sep=0pt,draw, rectangle, fit= (mat-1-2) (mat-1-3)] {};
      \node [inner sep=0pt,draw, rectangle, fit= (mat-2-2) (mat-2-3)] {};
      \path [->] (mat-2-4)   edge [dotted,thin,>=latex]  node {} (mat-3-4);
      \node [inner sep=0pt,draw, rectangle, fit= (mat-3-5) (mat-3-6)] {};
      \node [inner sep=0pt,draw, rectangle, fit= (mat-4-5) (mat-4-6)] {};
      \path [->] (mat-4-7)   edge [dotted,thin,>=latex]  node {} (mat-5-7);
      \node [inner sep=0pt,draw, rectangle, fit= (mat-5-8) (mat-5-10)] {};
    \end{tikzpicture}
  }
  
  \centerline{
    \begin{tikzpicture}
      \matrix (mat) [matrix of nodes,ampersand replacement=\&,row sep=0.5cm]{
            \& $b$ \& $b$ \&     \& $b$ \& $a$ \&     \& $a$ \&     \& $b$ \& $\mathbf{\modif{4}{q_4}}$ \& $a$ \&     \& $b$ \&\& ($\gamma(\modif{4}{q_4})=(\modif{6}{q_6},1)$)\\
           \& $b$ \& $b$ \&     \& $b$ \& $a$ \&     \& $a$ \& $\mathbf{\modif{6}{q_6}}$ \& $b$ \&     \& $a$ \&     \& $b$ \&\& ($\delta(\modif{6}{q_6},ba)=\modif{8}{q_8}$)\\
           \& $b$ \& $b$ \&     \& $b$ \& $a$ \&     \& $a$ \&     \& $b$ \&     \& $a$ \& $\mathbf{\modif{8}{q_8}}$ \& $b$ \&\& ($\gamma(\modif{8}{q_8})=(\modif{6}{q_6},2)$)\\
                      \& $b$ \& $b$ \&     \& $b$ \& $a$ \&     \& $a$ \&     \& $b$ \&     \& $a$ \& $\mathbf{\modif{6}{q_6}}$ \& $b$ \&\ \& ($\delta(q_6,b)=q_7$)\\\
           \& $b$ \& $b$ \&     \& $b$ \& $a$ \&     \& $a$ \&     \& $b$ \&     \& $a$ \&     \& $b$ \& $\mathbf{\modif{7}{q_7}}$ \& \ \ \ \ \ \ ($\modif{7}{q_7}\in F$ $\Rightarrow$ $w\in L(S)$)\\
      };
      \node [inner sep=0pt,draw, rectangle, fit= (mat-1-8) (mat-1-10)] {};
      \path [->] (mat-1-11)   edge [dotted,thin,>=latex]  node {} (mat-2-9);
      \node [inner sep=0pt,draw, rectangle, fit= (mat-2-10) (mat-2-12)] {};
      \node [inner sep=0pt,draw, rectangle, fit= (mat-3-10) (mat-3-12)] {};
       \node [inner sep=0pt,draw, rectangle, fit= (mat-4-14) (mat-4-15)] {};
      \path [->] (mat-3-13)   edge [dotted,thin,>=latex]  node {} (mat-4-13);
    \end{tikzpicture}
  }
  
\end{example}

  During the traversal of a QDS, the computation of the associated path needs to perform some shifts in the input window: if a transition $(p_j,x_j,p_{j+1})$ belongs to \modif{$\Gamma$}{$\gamma$}, a shift can be performed only if \textbf{(1)} there exist enough symbols in the input window, \textbf{(2)} there exist enough remaining symbols on the path, \textbf{(3)} these symbols match, and \textbf{(4)} for the last shift there is at least one symbol to be read after the  matching symbols. These constraints are formally defined in Definition \ref{def shiftable}.
  
  \begin{definition}\label{def shiftable}
    Let $t=((q_1,x_1,q_2),\ldots,(q_n,x_n,q_{n+1}))$ be a path of a QDS $S=(\Sigma,m, \Gamma, {\cal Q}$, $0$, $F$, $\delta,\gamma)$. The path $t$ is \emph{shiftable} if for any 
    edge $(q_j,l,q_{j+1})\in\gamma$ of the path $t$, 
  \begin{enumerate}
  \item[ \emph{\textbf{(1)}}] $m-l\leq j$,
  \item[ \emph{\textbf{(2)}}] $m-l\leq n-j$,
  \item[ \emph{\textbf{(3)}}] $\displaystyle\bigcdot_{i=j+1-(m-l)}^ {j-1}x_i= \bigcdot_{i=j+1}^ {j-1+(m-l)}x_i$,
    \item[ \emph{\textbf{(4)}}]  If $(q_i,x_i,q_{i+1})\in \delta$  for all $j< i\leq n$    then  $n+1-j>m-l$.

    \end{enumerate}
  Moreover, the \emph{$\Sigma$-label} $w=\displaystyle\bigcdot_{i=1}^ {n}y_i$ of the shiftable path $t$ is defined by
   
  $y_i=\left\{\begin{array}{ll}
\varepsilon &  \text{if } i\in [j,j+m-l-1]\mbox{ for } (q_j,l,q_{j+1})\in \gamma \\
x_i & \text{otherwise}
\end{array}
\right.$ 

\end{definition}
Notice that, by convention, we set $\displaystyle\bigcdot_{i=j}^kx_i=\varepsilon$ if $j>k$. It is the case for  \textbf{(3)} if $l=m-1$.
  We also consider that there exists a shiftable empty path $t=(q,\varepsilon, q)$ from any state $q$ to itself.
  
  \begin{example}\label{ex exp chemin decal}
    Let us consider the path $(t_1,\ldots , t_{11})$ labelled by $w=ba2ab1ba2aa$ of the QDS of Figure \ref{fig ex qds graph}.
This path is shiftable since  the conditions are checked for every $\gamma$-transition:
 \begin{enumerate}
\item for $t_3=(q_5,2,q_1)$, we have \emph{\textbf{(1)}} $m-l=1\leq j=3$, \emph{\textbf{(2)}} $m-l=1\leq n-j=8$, \emph{\textbf{(3)}} $w[3,2]=\varepsilon=w[4,3]$, \emph{\textbf{(4)}} $t_3$ is not the last $\gamma$-transition,
\item for $t_6=(q_4,1,q_6)$, we have \emph{\textbf{(1)}} $m-l=2\leq j=6$, \emph{\textbf{(2)}} $m-l=2\leq n-j=5$, \emph{\textbf{(3)}} $w[5,5]=b=w[7,7]$, \emph{\textbf{(4)}} $t_6$ is not the last $\gamma$-transition,
\item for $t_9=(q_8,2,q_6)$, we have \emph{\textbf{(1)}} $m-l=1\leq j=9$, \emph{\textbf{(2)}} $m-l=1\leq n-j=2$, \emph{\textbf{(3)}} $w[9,8]=\varepsilon=w[10,9]$, \emph{\textbf{(4)}} as $t_9$ is  the last $\gamma$-transition, $n+1-j=3>m-l=1$.
 \end{enumerate}
  \end{example}

 \modif{}{Finally, the notion of successful path is easily extensible to QDS once the notion of shiftability is stated}.
  \begin{definition}\label{def successful}
    Let $t=((q_1,x_1,q_2),\ldots,(q_n,x_n,q_{n+1}))$ be a path of a QDS $S=(\Sigma, m, \Gamma, {\cal Q},0,F,$ $\delta,\gamma)$. The path $t$ is \emph{successful} if 
  \begin{itemize}
  \item $t$ is shiftable,
  \item $q_1=0$,
  \item $q_{n+1}\in F$.
  \end{itemize}
\end{definition}


As for automata, the language recognized by a QDS can be defined with respect to the notion of \emph{successful path}
as stated by the next lemma and its corollaries.
\begin{lemma}
Let $S=(\Sigma, m, \Gamma, {\cal Q},0,F,\delta,\gamma)$ be a quasi-deterministic structure, $q_1\in Q_1$, $w\in \Sigma ^*$, $q\in {\cal Q}$. The two following conditions are equivalent:
\begin{itemize}
\item $\Delta (q_1,w)=q$
\item the word $w$ is the $\Sigma$-label of  a shiftable path from $q_1$ to $q$. 
\end{itemize}
\end{lemma}
\begin{proof}

 The proof  is done by \modif{induction}{recurrence} on the length of $w$. 
 \begin{enumerate}
   \item For $|w|\leq m-1$\\    
   \centerline{$\begin{array}{lll}\Delta(q_1,w)=q& \Leftrightarrow &\delta(q_1,w)=q\\
   &\Leftrightarrow&
  \mbox{there exists a shiftable path labelled by $w$ from $q_1$ to $q$}
  \end{array}$}
    \smallskip
   \item Let us now consider that $|w|>m$. We have $\Delta(q_1,w)= \Delta(q',w[l+1,|w|])$ where $\gamma(\delta(q,w[1,m-1]))=(q',l)$. By the induction hypothesis, $\Delta(q',w[l+1,|w|])=q$ if and only if the word $w[l+1,|w|]$ is the $\Sigma$-label of a shiftable path from $q'$ to $q$.    
 Let $((q',x_1,p_2)\ldots (p_s,x_s,q))$ be this path. Necessarily, the beginning of the label of this path  is $x_1\cdots x_{m-l-1}=w_{l+1}\cdots w_{m-1}$. 
 Hence, 
 $\Delta(q_1,w)=q$ $\Leftrightarrow$ $\Delta(q_1,w[1,m-1])=q_m \ \wedge \modif{\Gamma}{\gamma}(q_m)=(q',l)\ \wedge \Delta(q',w[l+1,|w|])=q$ 
 $\Leftrightarrow$ there exists a path  $(q_1,w[1],q_2)\ldots (q_{m-1},w[m-1],q_{m})$, a transition  $(q_{m},l,q')$ and a shiftable path from $q'$ to $q$ labelled by $w_{l+1}\cdots w_{m-1}$ 
 $\Leftrightarrow$ there exists a shiftable path from $q_1$ to $q$ labelled by $w$.
 \end{enumerate}
\end{proof}

\begin{corollary}\label{Cor-1}
A word is recognized by  a quasi-deterministic structure if and only if it is the $\Sigma$-label of a successful path.
\end{corollary}

Finally, let us show how to determine whether a given word is recognized by a given quasi-deterministic structure (see Example~\ref{ex membership prob qds}).

\begin{algorithm}[H]
  \caption{Membership Test for Quasi-Deterministic Structure}
  \label{algo qds}
  \begin{algorithmic}[1]
    \REQUIRE $S=(\Sigma, m, \Gamma, {\cal Q},0,F,\delta,\gamma)$ a quasi deterministic structure, $w$ a word in $\Sigma^*$ 
    \ENSURE Returns $w\in L(S)$ 
    \IF{$|w|\leq m-1$}
      \RETURN $\delta(0,w)\in F$
    \ENDIF
    \STATE $q$ $\leftarrow$ $0$
    \STATE $w'$ $\leftarrow$ $w$
    \WHILE{$|w'|> m-1\ \wedge\ q\neq\bot$}
      \STATE $(q,j)$ $\leftarrow$ $\gamma(\delta(q,w'[1,m-1]))$
      \STATE $w'$ $\leftarrow$ $w'[j+1,|w'|]$
    \ENDWHILE    
    \RETURN $q\neq\bot\ \wedge\ \delta(q,w')\in F$ 
  \end{algorithmic}
\end{algorithm}

\begin{proposition}\label{prop algo mb ok}
  \modif{}{Let $w$ be a word in $\Sigma^*$ and $S=(\Sigma, m, \Gamma, {\cal Q},0,F,\delta,\gamma)$ be a QDS.}
  Algorithm~\ref{algo qds} returns \modif{TRUE if and only if}{the boolean} $w\in L(S)$. Furthermore, its execution always halts, and is performed in time $O((m-1)\times \frac{|w|}{s})$, $s=\mathrm{min}\{j\mid \exists q\in Q_m, \gamma(q)=(j,p)\}$ and in space $O(1)$.\end{proposition}
\begin{proof}
  Correctness is trivially proved from Definition~\ref{def lang qds} and Definition~\ref{def ext delta}. Space Complexity is constant since the only informations needed are the current state and the next portion of the word. Finally, time complexity is due to the loop from line 6 to line 9: the shift in $w'$ is at least equal to $\mathrm{min}\{j\mid \exists q\in Q_{m-1}, \gamma(q)=(j,p)\}$ and the computation of $\delta(q,w[1,m-1])$ can be performed in $O(m-1)$.\modif{ Hence the announced complexity.}{}  
  
\end{proof}

\modif{N}{The n}ext section is devoted to the conversion of a $(k,l)$-unambiguous NFA into a quasi-deterministic structure.

\section{From a (k,l)-unambiguous NFA to a quasi-deterministic structure}\label{se:klnanfa2qds}

For any $(k,l)$-unambiguous automaton, given a state $q$ and a word $w$ of length $k$, there exists an integer $i\leq l$ \modif{such that there exists}{and} at most one state $q'$ in $\delta(q,w[1,i])$ such that $\delta(q',w[i+1,k])$ is not empty. \modif{}{Assume that such an $i$ is taken as large as possible.} The integer $i$ is called the \emph{step index of} $q$ with respect to $w$ and is denoted by  \emph{$\mathrm{StepIndex}_w(q)$}. The state $q'$ is called the \emph{step successor} of $q$ with respect to $w$ and is denoted by \emph{$\mathrm{StepSucc}_w(q)$}.

Quasi-deterministic structures can be used in order to simulate each run in a unique way. For any pair $(q,w)$, $\mathrm{StepIndex}_w(q)$  and $\mathrm{StepSucc}_w(q)$ can be precomputed; then the run can restart in $\mathrm{StepSucc}_w(q)$ with a word $w'$ that is a suffix of $w$.

  \begin{example}\label{ex:calculstep} 
    Let $\Sigma=\{a,b\}$. Let $A$ be the automaton of Figure~\ref{fig ex aut klna} that denotes the language $\Sigma^*\cdot \{a\}\cdot \Sigma$. It can be shown that the automaton $A$ is a $(3,1)$-unambiguous NFA. As an example let us consider the state $\modif{1}{q_1}$. For any word $w$ in $\Sigma^3$ we have:
      $$|
      \{q'\in Q : q'\in\delta(\modif{1}{q_1},w[1,1])\ \wedge\ \delta(q',w[2,3])\neq\emptyset\}
    |\leq 1,$$

  \noindent \emph{i.e.} for any word $w$ in $\Sigma^3$, $\mathrm{StepIndex}_w(\modif{1}{q_1})=1$ and $\mathrm{StepSucc}_w(\modif{1}{q_1})=\modif{1}{q_1}$.
  
%
%
%
%
   
  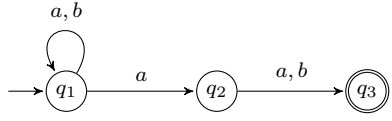
\begin{figure}[H]
    \centerline{  
      \begin{tikzpicture}[node distance=2cm,bend angle=30]   
	    \node[initial,state] (1) {$\modif{1}{q_1}$};  
	    \node[state] (2) [right of=1] {$\modif{2}{q_2}$};
	    \node[state,accepting] (3) [right of=2] {$\modif{3}{q_3}$};
	    \path[->]
	      (1)   edge  node {$a$} (2)
	      (2)   edge  node {$a,b$} (3)
	      (1)   edge [swap,in=120,out=60,loop] node {$a,b$} ();	            
      \end{tikzpicture}     
    }   
    \caption{The automaton $A$.}
    \label{fig ex aut klna}
  \end{figure}
  \end{example}


  The computation of the pairs $(\mathrm{StepIndex}_w(q),\mathrm{StepSucc}_w(q))$ for any pair $(q,w)$ of a state and a word is sufficient to compute a quasi-deterministic structure. 
  \modif{}{Indeed, for any state $q$, reading $k$ symbols $w=w_1\cdots w_k$ is enough to compute without ambiguity the unique successor $\mathrm{StepSucc}_w(q)$. Hence defining the states of the QDS as couples (state,word) is sufficient.}
  The quasi-deterministic structure is exponentially bigger with respect to the size of the alphabet than the automaton. That has to be compared with the exponential growth with respect to the number of states in the classical determinization.

\begin{definition}\label{def qds}
  Let $A=(\Sigma,Q,\{i\},F,\delta)$ be a $(k,l)$-unambiguous automaton. The \emph{quasi-deterministic structure associated with} $A$ is $S=(\Sigma, k+1, \Gamma, {\cal Q},$ $0,F',$ $\delta',\gamma')$ where:
  \begin{itemize}
  \item $\Gamma \subset [1,k]$
    \item $\forall j\in[1,k+1]$, $Q_j=\modif{\{(q,w) \mid q\in Q\ \wedge\ w\in\Sigma^{j-1}\}}{Q\times \Sigma^{j-1}}$,
    \item $0=(i,\varepsilon)$,
    \item $F'=\{(q,w) \mid \delta(q,w)\cap F\neq\emptyset\}$,
    \item 
    
    $\delta'((q,w),a)=
        \left\{
          \begin{array}{l@{\ }l}
            (q,w\cdot a) & \text{ if } \delta(q,w\cdot a)\in Q,\\
            \bot & \text{ otherwise,}\\
          \end{array}
        \right.$
    
    \item $\forall (q,w)\in Q_{k+1}$, $\gamma'((q,w))=(\mathrm{StepIndex}_w(q),(\mathrm{StepSucc}_w(q),\varepsilon))$.
  \end{itemize}
\end{definition}

\modif{}{Let us show that the QDS associated with any $(k,l)$-automaton $A$ is exponential w.r.t. the size of the alphabet and recognizes $L(A)$. First, as a direct consequence of Definition~\ref{def qds}, the following proposition holds.}

\begin{proposition}\label{prop taille qds associ}
  Let $A=(\Sigma,Q,I,F,\delta)$ be a $(k,l)$-unambiguous automaton and $S$ be the quasi-deterministic structure associated with $A$. Then the number of states of $S$ is $|Q| \times \frac{|\Sigma|^{k+1}-1}{|\Sigma|-1}$.
\end{proposition}
\modif{\begin{proof}
  Trivially according to Definition~\ref{def qds}.
  
\end{proof}}{}

\begin{proposition}\label{prop langage qds}
  Let $A$ be a $(k,l)$-unambiguous automaton and $S$ be the quasi-deterministic structure associated with $A$. Then\modif{:}{ $L(S)=L(A)$.}

  \centerline{\modif{$L(S)=L(A)$.}{}}
\end{proposition}
\begin{proof}
  Let $A=(\Sigma,Q,\{i\},F,\delta)$, $S=(\Sigma, k+1, \Gamma, \mathcal{Q},0,F',\delta',\gamma')$ and $\Delta$ be the extended transition function of $S$ (See Definition \ref{def ext delta}). Let $w$ be a word in $\Sigma^*$. Let us show by recurrence on the length of $w$ that:
  
  \centerline{
    \begin{tabular}{cp{3cm}}
      $\forall q\in Q$, $\Delta((q,\varepsilon),w)\in F'$ $\Leftrightarrow$ $\delta(q,w)\cap F\neq\emptyset$ &  \hfill(\textbf{P1})\\
    \end{tabular}
  }
  Suppose that $|w|\leq k$. We have $\Delta((q,\varepsilon),w)=\delta'((q,\varepsilon),w)$. According to Definition~\ref{def qds}, \textbf{(a)} $\delta'((q,\varepsilon),w)=(q,w)$ and \textbf{(b)} $(q,w)\in F'$ $\Leftrightarrow$ $\delta(q,w)\cap F\neq \emptyset$. (\textbf{P1}) is satisfied.
  
  Let us suppose that  (\textbf{P1}) is satisfied for any word $w$ such that  $k<|w|<n$. Suppose now that $|w|=n$. Either \textbf{(Case I)} $\Delta((q,\varepsilon),w)=\Delta((p,\varepsilon),w[j+1,n])$ if $\gamma'(\delta'((q,\varepsilon),w[1,k]))=(j,(p,\varepsilon)) \ \wedge\ p\neq\bot$ or \textbf{(Case II)} $\Delta((q,\varepsilon),w)=\bot$. \textbf{Case II} implies that $\delta(q,w)=\emptyset$ and consequently $w$ is neither in $L(A)$ nor in $L(S)$. Suppose that \textbf{Case I} holds. By the recurrence hypothesis, $\Delta((p,\varepsilon),w[j+1,n])\in F'$ $\Leftrightarrow$ $\delta(p,w[j+1,n])\cap F\neq\emptyset$. Since $\gamma\modif{}{'}(\delta'((q,\varepsilon),w[1,k]))=(j,(p,\varepsilon))$, 
 we have    $\mathrm{StepIndex}_{w[1,k]}((q,\varepsilon))=j$ and $\mathrm{StepSucc}_{w[1,k]}((q,\varepsilon))=(p,\varepsilon)$ which implies $p\in \delta(q,w[1,j])$. As a consequence $\delta(p,w[j+1,n])\cap F\neq\emptyset$ $\Leftrightarrow$ $\delta(q,w)\cap F\neq\emptyset$. Finally $\Delta((q,\varepsilon),w)\in F'$ $\Leftrightarrow$ $\delta(q,w)\cap F\neq\emptyset$ and (\textbf{P1}) holds.
  
  As a conclusion, (\textbf{P1}) holds for $q=i$ and since for all $w$ in $\Sigma^*$, $\Delta((i,\varepsilon),w)\in F'$ $\Leftrightarrow$ $\delta(i,w)\cap F\neq\emptyset$, equality of languages holds.
  
\end{proof}

  
  \begin{example}\label{ex:quasidetstruct}
   Let us consider the automaton $A$ defined in Example~\ref{ex:calculstep}. After removing unreachable states, the quasi-deterministic structure associated with $A$ is given in Figure~\ref{fig qdet asso ex}.
  \end{example}
  
  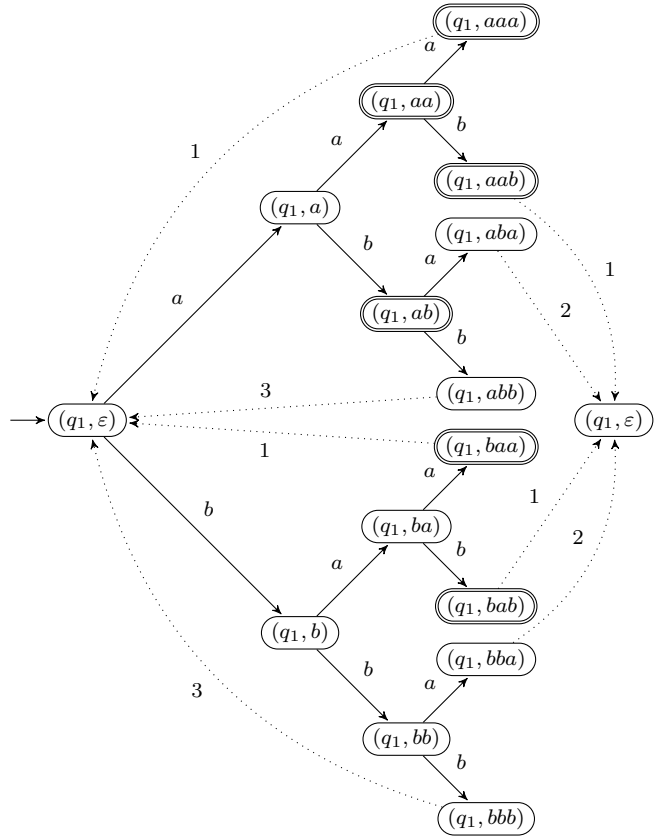
\begin{figure}[H]
    \centerline{  
      \begin{tikzpicture}[node distance=2cm,bend angle=30]   
	    \node[initial,state,rounded rectangle] (1eps) {$(\modif{1}{q_1},\varepsilon)$}; 
	    \node[state,rounded rectangle,node distance=7cm,right of=1eps] (1eps2) {$(\modif{1}{q_1},\varepsilon)$};   
	    \node[node distance=4cm,state,rounded rectangle,above right of=1eps] (1a) {$(\modif{1}{q_1},a)$};  
	    \node[node distance=4cm,state,rounded rectangle,below right of=1eps] (1b) {$(\modif{1}{q_1},b)$};  
	    \node[state,rounded rectangle,above right of=1a,accepting] (1aa) {$(\modif{1}{q_1},aa)$};   
	    \node[state,rounded rectangle,below right of=1a,accepting] (1ab) {$(\modif{1}{q_1},ab)$};
	    \node[node distance=1.5cm,state,rounded rectangle,above right of=1aa,accepting] (1aaa) {$(\modif{1}{q_1},aaa)$};   
	    \node[node distance=1.5cm,state,rounded rectangle,below right of=1aa,accepting] (1aab) {$(\modif{1}{q_1},aab)$};
	    \node[node distance=1.5cm,state,rounded rectangle,above right of=1ab] (1aba) {$(\modif{1}{q_1},aba)$};   
	    \node[node distance=1.5cm,state,rounded rectangle,below right of=1ab] (1abb) {$(\modif{1}{q_1},abb)$};
	    \node[state,rounded rectangle,above right of=1b] (1ba) {$(\modif{1}{q_1},ba)$};   
	    \node[state,rounded rectangle,below right of=1b] (1bb) {$(\modif{1}{q_1},bb)$};	    \node[node distance=1.5cm,state,rounded rectangle,above right of=1ba,accepting] (1baa) {$(\modif{1}{q_1},baa)$};   
	    \node[node distance=1.5cm,state,rounded rectangle,below right of=1ba,accepting] (1bab) {$(\modif{1}{q_1},bab)$};
	    \node[node distance=1.5cm,state,rounded rectangle,above right of=1bb] (1bba) {$(\modif{1}{q_1},bba)$};   
	    \node[node distance=1.5cm,state,rounded rectangle,below right of=1bb] (1bbb) {$(\modif{1}{q_1},bbb)$}; 
	    \path[->]
	      (1eps)   edge  node {$a$} (1a)
	      (1eps)   edge  node {$b$} (1b)
	      (1a)   edge  node {$a$} (1aa)
	      (1a)   edge  node {$b$} (1ab)
	      (1b)   edge  node {$a$} (1ba)
	      (1b)   edge  node {$b$} (1bb)
	      (1aa)   edge  node {$a$} (1aaa)
	      (1aa)   edge  node {$b$} (1aab)
	      (1ab)   edge  node {$a$} (1aba)
	      (1ab)   edge  node {$b$} (1abb)
	      (1ba)   edge  node {$a$} (1baa)
	      (1ba)   edge  node {$b$} (1bab)
	      (1bb)   edge  node {$a$} (1bba)
	      (1bb)   edge  node {$b$} (1bbb)
	      (1aaa)   edge[dotted,bend right,swap]  node {$1$} (1eps)
	      (1aab)   edge[dotted,bend left]  node {$1$} (1eps2)
	      (1baa)   edge[dotted]  node {$1$} (1eps)
	      (1bab)   edge[dotted]  node {$1$} (1eps2)
	      (1aba)   edge[dotted]  node {$2$} (1eps2)
	      (1bba)   edge[dotted,bend right]  node {$2$} (1eps2)
	      (1abb)   edge[dotted,swap]  node {$3$} (1eps)
	      (1bbb)   edge[dotted,bend left]  node {$3$} (1eps)
	      ;	            
      \end{tikzpicture}     
    }   
    \caption{The Quasi-Deterministic Structure Associated with $A$.}
    \label{fig qdet asso ex}
  \end{figure}

\section{Reduction of a quasi-deterministic structure}\label{sec qds}
\modif{}{In this section, we show how to reduce the number of states in a QDS, first by getting rid of  useless states, then by merging equivalent states.}
\subsection{Accessibility and co-accessibility of a QDS}
The definition of trim quasi-deterministic structure differs from the one of trim automaton. Indeed, accessibility and
co-accessibility as defined in automata are not enough 
to obtain
a trim quasi-deterministic structure (see Example~\ref{ex useful} for an illustration). 

\begin{definition}
\modif{}{A state of a QDS is useful if it is on a successful path or if it is initial. A transition is useful if it appears on a successful path. The finality of a state is useful if this state is the destination of a successful path. A QDS is trim if \textbf{(1)}  each state is useful, \textbf{(2)} each transition is useful and \textbf{(3)} the finality of each final state is useful.} 
\end{definition}

\begin{example}\label{ex useful}
  Consider the QDS \modif{}{$(\Sigma, m=3, \Gamma,{\cal  Q},q_1,F,\delta,\gamma)$} in Figure~\ref{fig ex qds trim}. The only useful states are states $\modif{1}{q_1}$ and $\modif{2}{q_2}$. Indeed, states $\modif{7}{q_7}$ and $\modif{8}{q_8}$ in Figure~\ref{fig ex qds trim} are not on any shiftable path: using the $\gamma$-transition of label $1$, there is  a symbol $b$ or $c$ in the reading window when the QDS reaches the state $\modif{4}{q_4}$. Hence there is no word $w$ such that $\Delta(q_1,w)\in\{\modif{7}{q_7},\modif{8}{q_8}\}$. Furthermore, since there is no word $w$ such that $\Delta(q_1,w)=\modif{5}{q_5}$, and as $q_6$ is not a final state, both $\modif{5}{q_5}$ and $q_6$ are not useful.
\modif{}{Indeed, $\Delta(q_1,ab)=q_3$, $\Delta(q_1,aba)=\bot$ and for any integer $j>0$, $\Delta(q_1,abb^j)=q_6$.}
\end{example}  

\begin{minipage}[b]{0.49\linewidth}
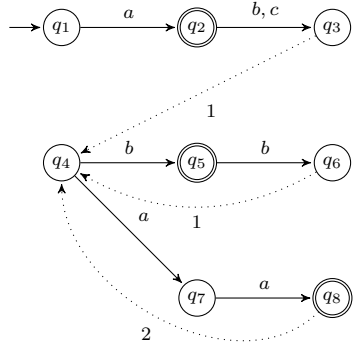
\begin{figure}[H]
  \centerline{
    \begin{tikzpicture}[node distance=2cm,bend angle=30,transform shape,scale=0.9]   
	    \node[initial, state] (1) {$\modif{1}{q_1}$};     
	    \node[accepting, state, right of=1] (2) {$\modif{2}{q_2}$}; 
	    \node[state, right of=2] (4) {$\modif{3}{q_3}$}; 
	    \node[state, below of=1] (6) {$\modif{4}{q_4}$};  
	    \node[accepting, state, right of=6] (7) {$\modif{5}{q_5}$}; 
	    \node[state, right of=7] (8) {$\modif{6}{q_6}$};  
	    \node[state, below of=7] (9) {$\modif{7}{q_7}$};  
	    \node[state, right of=9,accepting] (10) {$\modif{8}{q_8}$};  
	    \path[->]
	      (1)   edge  node {$a$} (2)
	      (2)   edge  node {$\modif{b}{b,c}$} (4)
	      (6)   edge  node {$b$} (7)
	      (6)   edge  node {$a$} (9)
	      (7)   edge  node {$b$} (8)
	      (9)   edge  node {$a$} (10)
	      (4)   edge[dotted]  node {$1$} (6)
	      (8)   edge[dotted,bend left]  node {$1$} (6)
	      (10)   edge[in=-90,out=-135,dotted,looseness=1]  node {$2$} (6)
	      ;	            
      \end{tikzpicture}}
  \caption{A Not-trim Quasi-Deterministic Structure.}
  \label{fig ex qds trim}
\end{figure} 
 \end{minipage}  
\begin{minipage}[b]{0.49\linewidth}
\begin{figure}[H]
  \centerline{
    \begin{tikzpicture}[node distance=2cm,bend angle=30,transform shape,scale=0.9]   
	    \node[initial, state] (1) {$\modif{1}{q_1}$};     
	    \node[accepting, state, right of=1] (2) {$\modif{2}{q_2}$}; 
	     \node[below of=1] (6){};
	    \path[->]
	      (1)   edge  node {$a$} (2);
      \end{tikzpicture}}
  \caption{Its trim part}
\end{figure}
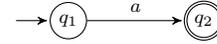 
 \end{minipage}  
\bigskip

\modif{A state is useful if it appears on a successful path or if it is initial. A transition is useful if it appears on a successful path. The finality of a state is useful if this state is the destination of a successful path.
 A trim part of a quasi-deterministic structure can be computed by keeping only useful components. }{ Let us show that the trim part 
 of a QDS is computable.}
Let $w,w'$ be two words of $\Sigma^*$. We denote by $w'\leq w$ (resp $w'<w$)  if $w'$ is a prefix of $w$ ($w'$ is a proper prefix of $w$). 

In order to compute the trim part of a QDS, we need to decide whether a state, an edge or a final state appears on a successful path.  The successful  paths can be computed through the path-DFA associated with any QDS, defined as follows.

\begin{definition}
  Let $S=(\Sigma,m, \Gamma, {\cal Q},0,F,\delta,\gamma)$ be a QDS. The path-DFA of $S$ is the DFA $(\Sigma',Q',$ $\{i'\},F',\delta')$ defined as follows:
  \begin{itemize}
    \item $\Sigma'=\Sigma \cup \Gamma$ 
    \item $Q'=\{(p,u,v)\in Q\times \Sigma^{l_1}\times \Sigma^{l_2}\mid l_1,l_2\leq m-1\wedge (u\leq v \vee v\leq u)\}$
    \item $i'=(0,\varepsilon,\varepsilon)$
    \item $F'=\{(p,u,v)\in Q'\mid p\in F \wedge v<u \}$
    \item $\delta'((p,u,v),a)=
      \left\{
        \begin{array}{ll}
          (\delta(p,a),ua,v)&\text{ if }p\not\in Q_m,\ a\in \Sigma, \delta(p,a)\neq \bot, \\
          & \ u\cdot a\leq v \vee v\leq u,\\
          (q,\varepsilon,u[l+1,|u|])&\text{ if } p\in Q_m,\gamma(p)=(q,l), a=l,\\
          \bot&\text{ otherwise}
        \end{array}
      \right.$
  \end{itemize}
\end{definition}

  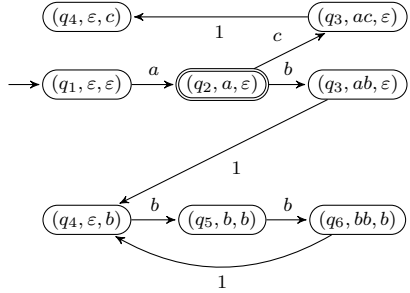
\begin{figure}[H]
  \centerline{
    \begin{tikzpicture}[node distance=2cm,bend angle=30,transform shape,scale=0.9]   
	    \node[initial, state, rounded rectangle] (1) {$(\modif{1}{q_1},\varepsilon,\varepsilon)$};   
	    \node[ state, rounded rectangle, above of =1,node distance=1cm] (6b) {$(\modif{4}{q_4},\varepsilon,c)$};     
	    \node[accepting, state, right of=1, rounded rectangle] (2) {$(\modif{2}{q_2},a,\varepsilon)$}; 
	    \node[state, right of=2, rounded rectangle] (4) {$(\modif{3}{q_3},ab,\varepsilon)$}; 
	    \node[state, above of=4, rounded rectangle, node distance=1cm] (4b) {$(\modif{3}{q_3},ac,\varepsilon)$}; 
	    \node[state, below of=1, rounded rectangle] (6) {$(\modif{4}{q_4},\varepsilon,b)$};  
	    \node[ state, right of=6, rounded rectangle] (7) {$(\modif{5}{q_5},b,b)$}; 
	    \node[state, right of=7, rounded rectangle] (8) {$(\modif{6}{q_6},bb,b)$}; 
	    \path[->]
	      (1)   edge  node {$a$} (2)
	      (2)   edge  node {$b$} (4)
	      (2)   edge  node {$c$} (4b)
	      (6)   edge  node {$b$} (7)
	      (7)   edge  node {$b$} (8)
	      (4)   edge  node {$1$} (6)
	      (4b)   edge  node {$1$} (6b)
	      (8)   edge[bend left]  node {$1$} (6)
	      ;	            
      \end{tikzpicture}  
  }
  \caption{The  Path-DFA of the QDS of Figure \ref{fig ex qds trim}.}
\end{figure}

  \modif{}{Let us now explicit the relation between the shiftable paths of a QDS and a path in the associated path-DFAs.}

\begin{theorem}\label{link-QDS-pathDFA}
  Let  $S=(\Sigma,m, \Gamma, {\cal Q},0,F,\delta,\gamma)$ be a QDS. Let $A=(\Sigma',Q',\{i'\},$ $F',\delta')$ be the path-DFA of $S$. The three following conditions hold:
  \begin{enumerate}
    \item A state $p\in {\cal Q}$ is useful  if and only if there exists a useful state $(p,u,v)$ in $Q'$.
    \item A transition $(p,a,q)\in \delta\cup \gamma$ is useful if and only if there exists a transition $((p,u,v),a,(q,u',v'))$ in $\delta'$ where the states $(p,u,v)$ and $(q,u',v')$ are useful.
    \item The finality of the state $p\in  {\cal Q}$ is useful if and only if there exists a final useful state $(p,u,v)$ in $Q'$.
  \end{enumerate}
\end{theorem}

The notion of successful path in a QDS can be expressed through the notion of successful path in its associated path-DFA.

\begin{lemma}\label{lem-shift}
 Let $S=(\Sigma, m, \Gamma ,{\cal Q},0,F,\delta,\gamma)$ be a QDS and $A=(\Sigma \cup \Gamma,Q',\{(0,\varepsilon,\varepsilon)\},F',\delta')$ be its path-DFA. Let $w=a_1\cdots a_n$ be a word over $\Sigma \cup \Gamma$ and $t=((0,a_1,p_1)\ldots (p_{n-1},a_n,p_n))$ be a path in $S$. The two following conditions are equivalent:
  \begin{enumerate}
    \item The path $t$ is shiftable,
    \item  there exists two words $u,v\in \Sigma^*$ with $v< u$ and a path $t'$ in $A$ labelled $w$ from $(0,\varepsilon,\varepsilon)$ to $(p_n,u,v)$.
  \end{enumerate}
  Moreover, $t$ is successful in $S$ if and only if $t'$ is successful in $A$.
\end{lemma}
\begin{proof}
\ \\
\begin{itemize}
\item[$\mathbf{(1)\Rightarrow (2)}$] By recurrence on the number $k$ of $\gamma$-transitions in $t$. If $k=0$, then $w=a_1\cdots a_n$ is in $\Sigma^*$. By construction, there exists a path $t'$ in $A$ labelled $w$ from $(0,\varepsilon, \varepsilon)$ to $(p_n,w,\varepsilon)$. Let us show the property for a path with $k+1$ $\gamma$-transitions. Then $t=((0,a_1,p_1)\cdots (p_{i-1},a_i,p_i)$ $(p_i,l,p_{i+1})(p_{i+1},a_{i+2},p_{i+2}) \cdots (p_{n-1},a_n,p_n))$ with $(p_i,l,p_{i+1})$ the $k+1^{th}$ $\gamma$-transition of $t$. By definition of a shiftable path, $t=((0,a_1,p_1)\cdots (p_{i-1},a_i,p_i))$ is shiftable and then by the recurrence hypothesis there exists a path in $A$ from $(0,\varepsilon,\varepsilon)$ to $(p_i,u,v)$ with $v < u$.
By construction of $A$, $\delta '((p_i,u,v),l)=((p_{i+1},\epsilon,u[l+1,|u|]))$. As $t$ is shiftable, we have $a_{i+2}\cdots a_{i+2+|u|-(l+1)}=u[l+1,|u|]$. So, by construction, there exists a path in $A$ from $(p_{i+1},\epsilon,u[l+1,|u|])$ to $(p_n,a_{i+2}\cdots a_n,u[l+1,|u|])$ with $a_{i+2}\cdots a_n$ in $\Sigma^*$ and $u[l+1,|u|]< a_{i+2}\cdots a_n$ which ends the proof.
\item[$\mathbf{(2)\Rightarrow (1)}$] By recurrence on the number $k$ of transitions labelled by a symbol of $\Gamma$ in $t'$ ($\Gamma$-transitions). If $k=0$ then $t'=(((0,\varepsilon,\varepsilon),a_1,(p_1,a_1,\varepsilon))\cdots ((p_{n-1},a_1\cdots a_{n-1},\varepsilon),a_n,(p_n,w,\varepsilon)))$. Consequently, the path $t=((0,a_1,p_1)\cdots (p_{n-1},a_n,p_n))$ of $S$ is shiftable. Let us show the property for a path with $k+1$ $\Gamma$-transitions. Let $t'=(((0,\varepsilon,\varepsilon),a_1,(p_1,a_1,\varepsilon))\cdots ((p_{i-1},u',v'),a_i,(p_i,$ $u,v))((p_i,u,v),l,(p_{i+1},$ $\varepsilon,u[l+1,|u|]))\cdots ((p_{n-1},u'_1,v'_1),a_n,(p_n,u_1,v_1)))$ with $((p_i,u,v), l, (p_{i+1},$ $\varepsilon,u[l+1,|u|]))$ the last $\Gamma$-transition of $t'$. By the induction hypothesis  $v< u$ implies that the path $((0,a_1,p_1)\cdots (p_{i-1},a_i,p_i))$ is a shiftable path in $S$. As there exists a path from $(p_{i+1},\varepsilon, u[l+1,|u|])$ to $(p_n,u_1,v_1)$ in $A$ with $u_1\in \Sigma^*$, $v_1=u[l+1,|u|]< u_1$. So $u[l+1,|u|]=a_{i-(|u|-(l+1))}\cdots a_i=a_{i+2}\cdots a_{i+2+(|u|-(l+1))}<a_{i+2}\cdots a_n$ and so $t$ is a shiftable path.
\end{itemize}
\end{proof}
\begin{proof}(of Theorem \ref{link-QDS-pathDFA})
Conditions $1$, $2$ and $3$ directly holds  from Definition \ref{def successful} and Lemma \ref{lem-shift}.
\end{proof}

\begin{example}\label{ex path DFA}
  Consider the QDS in Figure~\ref{fig ex qds path dfa} \modif{}{(the  QDS of  Figure~\ref{fig ex qds trim} in which the state $q_6$ has been made final)}.  The states $\modif{7}{q_7}$ and $\modif{8}{q_8}$ are not accessible, the transition $(\modif{2}{q_2},c,\modif{3}{q_3})$ is not useful, and neither is the finality of $\modif{5}{q_5}$. The accessible part of its path-DFA is given Figure~\ref{fig ex dfa path dfa} and its trim part in Figure~\ref{fig ex trim qds path dfa}.
\end{example}  

\begin{minipage}[b]{0.49\linewidth}
\begin{figure}[H]
  \centerline{
    \begin{tikzpicture}[node distance=2cm,bend angle=30,transform shape,scale=0.9]   
	    \node[initial, state] (1) {$\modif{1}{q_1}$};     
	    \node[accepting, state, right of=1] (2) {$\modif{2}{q_2}$}; 
	    \node[state, right of=2] (4) {$\modif{3}{q_3}$}; 
	    \node[state, below of=1] (6) {$\modif{4}{q_4}$};  
	    \node[accepting, state, right of=6] (7) {$\modif{5}{q_5}$}; 
	    \node[accepting,state, right of=7] (8) {$\modif{6}{q_6}$};  
	    \node[state, below of=7] (9) {$\modif{7}{q_7}$};  
	    \node[state, right of=9,accepting] (10) {$\modif{8}{q_8}$};  
	    \path[->]
	      (1)   edge  node {$a$} (2)
	      (2)   edge  node {$b,c$} (4)
	      (6)   edge  node {$b$} (7)
	      (6)   edge  node {$a$} (9)
	      (7)   edge  node {$b$} (8)
	      (9)   edge  node {$a$} (10)
	      (4)   edge[dotted]  node {$1$} (6)
	      (8)   edge[dotted,bend left]  node {$1$} (6)
	      (10)   edge[in=-90,out=-135,dotted,looseness=1]  node {$2$} (6)
	      ;	            
      \end{tikzpicture}  
  }
  \caption{A Quasi-Deterministic Structure $S$.}
  \label{fig ex qds path dfa}
\end{figure}
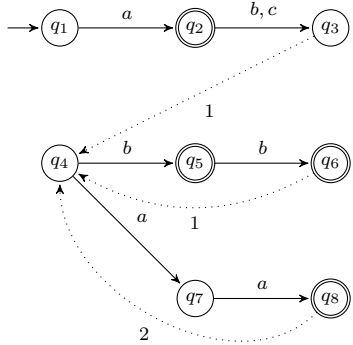 
\end{minipage}
\begin{minipage}[b]{0.49\linewidth}
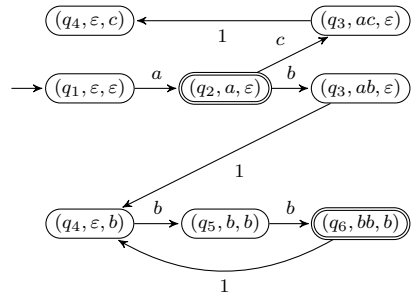
\begin{figure}[H]
  \centerline{
    \begin{tikzpicture}[node distance=2cm,bend angle=30,transform shape,scale=0.9]   
	    \node[initial, state, rounded rectangle] (1) {$(\modif{1}{q_1},\varepsilon,\varepsilon)$};   
	    \node[ state, rounded rectangle, above of =1,node distance=1cm] (6b) {$(\modif{4}{q_4},\varepsilon,c)$};     
	    \node[accepting, state, right of=1, rounded rectangle] (2) {$(\modif{2}{q_2},a,\varepsilon)$}; 
	    \node[state, right of=2, rounded rectangle] (4) {$(\modif{3}{q_3},ab,\varepsilon)$}; 
	    \node[state, above of=4, rounded rectangle, node distance=1cm] (4b) {$(\modif{3}{q_3},ac,\varepsilon)$}; 
	    \node[state, below of=1, rounded rectangle] (6) {$(\modif{4}{q_4},\varepsilon,b)$};  
	    \node[ state, right of=6, rounded rectangle] (7) {$(\modif{5}{q_5},b,b)$}; 
	    \node[accepting,state, right of=7, rounded rectangle] (8) {$(\modif{6}{q_6},bb,b)$}; 
	    \path[->]
	      (1)   edge  node {$a$} (2)
	      (2)   edge  node {$b$} (4)
	      (2)   edge  node {$c$} (4b)
	      (6)   edge  node {$b$} (7)
	      (7)   edge  node {$b$} (8)
	      (4)   edge  node {$1$} (6)
	      (4b)   edge  node {$1$} (6b)
	      (8)   edge[bend left]  node {$1$} (6)
	      ;	            
      \end{tikzpicture}  
  }
  \caption{The Path-DFA of S.}
  \label{fig ex dfa path dfa}
\end{figure}  
\end{minipage}

\begin{minipage}[b]{0.49\linewidth}
\begin{figure}[H]
  \centerline{
    \begin{tikzpicture}[node distance=2cm,bend angle=30,transform shape,scale=0.9]   
	    \node[initial, state, rounded rectangle] (1) {$(\modif{1}{q_1},\varepsilon,\varepsilon)$};   
	    \node[accepting, state, right of=1, rounded rectangle] (2) {$(\modif{2}{q_2},a,\varepsilon)$}; 
	    \node[state, right of=2, rounded rectangle] (4) {$(\modif{3}{q_3},ab,\varepsilon)$}; 
	    \node[state, below of=1, rounded rectangle] (6) {$(\modif{4}{q_4},\varepsilon,b)$};  
	    \node[ state, right of=6, rounded rectangle] (7) {$(\modif{5}{q_5},b,b)$}; 
	    \node[accepting,state, right of=7, rounded rectangle] (8) {$(\modif{6}{q_6},bb,b)$}; 
	    \path[->]
	      (1)   edge  node {$a$} (2)
	      (2)   edge  node {$b$} (4)
	      (6)   edge  node {$b$} (7)
	      (7)   edge  node {$b$} (8)
	      (4)   edge  node {$1$} (6)
	      (8)   edge[bend left]  node {$1$} (6)
	      ;	            
      \end{tikzpicture}  
  }
  \caption{The Trim  Path-DFA.}
  \label{fig trim ex dfa path dfa}
\end{figure}
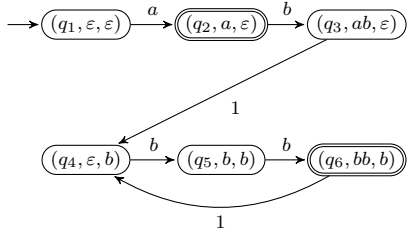  
\end{minipage}
\begin{minipage}[b]{0.49\linewidth}
\begin{figure}[H]
  \centerline{
    \begin{tikzpicture}[node distance=2cm,bend angle=30,transform shape,scale=0.9]   
	    \node[initial, state] (1) {$\modif{1}{q_1}$};     
	    \node[accepting, state, right of=1] (2) {$\modif{2}{q_2}$}; 
	    \node[state, right of=2] (4) {$\modif{3}{q_3}$}; 
	    \node[state, below of=1] (6) {$\modif{4}{q_4}$};  
	    \node[state, right of=6] (7) {$\modif{5}{q_5}$}; 
	    \node[accepting,state, right of=7] (8) {$\modif{6}{q_6}$};  
	    \path[->]
	      (1)   edge  node {$a$} (2)
	      (2)   edge  node {$b$} (4)
	      (6)   edge  node {$b$} (7)
	      (7)   edge  node {$b$} (8)
	      (4)   edge[dotted]  node {$1$} (6)
	      (8)   edge[dotted,bend left]  node {$1$} (6)
	      ;	            
      \end{tikzpicture}  
  }
  \caption{The Trim version of $S$.}
  \label{fig ex trim qds path dfa}
\end{figure}
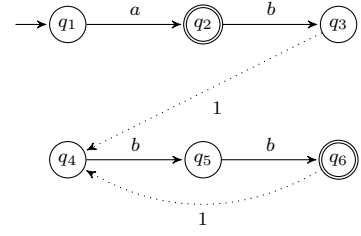 
\end{minipage}

\modif{}{As a direct consequence of Corollary~\ref{Cor-1}, we have:}

 \begin{lemma}
 Let $S'$ be the trim part of a QDS $S$. Then $L(S')=L(S)$.
 \end{lemma}
 \modif{\begin{proof}
 The proof is a direct consequence of Corollary~\ref{Cor-1}
 \end{proof}}{}
 
%

%
%
%
%

 In the following, we will consider trim QDS.

\subsection{Defining an equivalence relation $\sim$ for a QDS}

The minimization process in a DFA consists in merging each 
couple of states having the same right language.
The \emph{right language of a state} $q$ is defined as the set of words $w$ such that $\delta(q,w)$ is a final state. This computation leads to the canonical minimal DFA.

In a quasi-deterministic structure, the notion of right language can not be defined as in an automaton. Indeed, even if a state is useful, it can be not accessible. 
Thus, we \modif{can}{must} define an equivalence relation allowing us to reduce the size of a QDS. Notice that this
equivalence is sufficient
to reduce the number of states, but not to minimize: modifying the values of some $\modif{\Gamma}{\gamma}$-transitions might preserve the language while producing equivalent states.

Let $S=(\Sigma,m,\Gamma, {\cal Q}, 0,F,\delta,\gamma)$ be a QDS. The reduced QDS is always computed from a QDS not recognizing the empty word. Obviously, if $S$ recognizes $\varepsilon$, necessarily $0\in F$. In this case, $0$ is removed from $F$ to compute the reduced QDS $S'$ recognizing $L(S)\setminus \{\varepsilon\}$. The initial state of $S'$ is then made final to have $L(S')=L(S)$. Indeed, the finality of each final state of $Q_1$ is not useful, excepted for $0$ if $\varepsilon\in L(S)$.

We define an equivalence relation $\sim$ over ${\cal Q}=\bigcup_{k\in[1,m]}Q_k$. Let $p\in Q_i$, $q\in Q_j$.
  
 $$p\sim q \Leftrightarrow 
 \left\{\begin{array}{ll}
\text{(1)}& i=j,\\
\text{(2)}&p\in F \Leftrightarrow q\in F,\\
 \text{(3)}&\left\{\begin{array}{lll}
 \delta(p,a)\sim \delta(q,a)&\text{ if }i\neq m\\
 \tilde{\gamma}(p)= \tilde{\gamma}(q) \wedge  \overline{\gamma}(p)\sim \overline{\gamma}(q)&\text{ otherwise}
 \end{array}\right.
 \end{array}\right.
 $$ 
  
 \noindent
  We denote $[p]_{\sim}=\{q\mid p\sim q\}$.\\

  The \emph{quotient of} $S$ with respect to $\sim$ is the QDS $S_{\sim}=(\Sigma, m, \Gamma, \bigcup_{k\in[1,m]}Q_{k_{\sim}},[0]_\sim,F_{\sim},\delta_{\sim},\gamma_{\sim})$ defined by:
  \begin{itemize}
    \item $\forall k\in[1,m]$, ${Q_k}_{\sim}=\{[p]_{\sim}: p\in Q_k\}$,
      \item $F_{\sim}=\{[p]_{\sim}: [p]_{\sim}\cap F\neq \emptyset\}$, 
\item $\forall k\in[1,m-1]$, $\forall (p,a)\in Q_k\times\Sigma$, $\delta_{\sim}([p]_{\sim},a)=[\delta(p,a)]_\sim$,
    \item $\forall p\in Q_m$, $\gamma_{\sim}([p]_{\sim})=(l,[q]_\sim)$ with $\gamma(p)=(l,q)$.
  \end{itemize}  
  \modif{where for any state $q$, $[q]_{\sim}=\{q'\mid q'\sim q\}$.}{}
  
  We first show that merging two equivalent states preserves the language. 
  This property is sufficient, but not necessary. 

\begin{proposition}\label{prop rightinv pres lang}
  Let $S=(\Sigma,m,\Gamma, \bigcup_{k\in[1,m]}Q_k, 0,F,\delta,\gamma)$ be a QDS and $S_{\sim}=(\Sigma, m, \Gamma, \bigcup_{k\in[1,m]}Q_{k_{\sim}},$ $[0]_\sim,F_{\sim},\delta_{\sim},\gamma_{\sim})$  be the quotient of $S$ with respect to $\sim$. We have:
  
  \centerline{
    $L(S)=L(\modif{S'}{S_{\sim}})$.
  }
\end{proposition}
\begin{proof}
Let $\Delta_\sim$ be the extended transition function of $S_{\sim}$.
  \begin{enumerate}
    \item Let us consider the empty word $\varepsilon$:
  
  \centerline{
      $\varepsilon\in L(S)$  $\Leftrightarrow$  $0\in F$  $\Leftrightarrow$  $[0]_\sim \in \modif{F'}{F_\sim}$ \modif{$\Leftrightarrow$  $i' \in F'$}{}  $\Leftrightarrow$  $\varepsilon\in L(S_\sim)$
  }
  
  \item Let $w$ be a word in $\Sigma^+$. Let us show by recurrence on the length of $w$ that for any element $q$ in $Q_1$, $\Delta(q,w)\in F$ $\Leftrightarrow$ $\modif{\Delta'}{\Delta_\sim}([q]_\sim,w)\in \modif{F'}{F_\sim}$.
  
   \begin{enumerate}
     \item Suppose that $|w|\leq m$. Then:
  
    \begin{tabular}{ll@{\ }l}
        $\Delta(q,w)\in F$ &
         $\Leftrightarrow$&  $\delta(q,w)\in F$ 
        $\Leftrightarrow$  $\exists f\in F\text{ such that } \delta(q,w)=f$\\
       & $\Leftrightarrow$  &$\exists [f]_{\sim}\in \modif{F'}{F_\sim}\text{ such that } \delta_\sim([q]_{\sim},w)=[f]_{\sim}$
        $\Leftrightarrow$  $\modif{\Delta'}{\Delta_\sim}([q],w)\in \modif{F'}{F_\sim}$.\\
    \end{tabular}
  
 \item Suppose that $|w|>m$. Then: 

    \begin{tabular}{l@{\ }l@{\ }l}
       $\Delta(q,w)\in F$
      &  $\Leftrightarrow$ & $\gamma(\delta(q,w[1,m]))=(l,p) \wedge \Delta(p,w[l+1,|w|])\in F$\\
      & $\Leftrightarrow$ &  $ \delta(q,w[1,m])=r \wedge \gamma(r)=(l,p) \wedge \Delta(p,w[l+1,|w|])\in F$\\
      & $\Leftrightarrow$ & $\delta_\sim([q]_\sim,w[1,m])=[r]_\sim$ $\wedge \modif{\gamma'}{\gamma_\sim}([r]_\sim)=(l,[p]_\sim) \wedge \modif{\Delta'}{\Delta_\sim}([p]_\sim,w[l+1,|w|])\in \modif{F'}{F_\sim}$\\
      & & (According to  \modif{}{the} recurrence hypothesis)\\
      & $\Leftrightarrow$ & $\modif{\Delta'}{\Delta_\sim}([q]_\sim,w)\in \modif{F'}{F_\sim}$\\
    \end{tabular}
  \end{enumerate}
  
  \item Finally, since for any element $q$ in $Q_1$, for any word $w$ in $\Sigma^*$, $\Delta(q,w)\in F$ $\Leftrightarrow$ $\modif{\Delta'}{\Delta_\sim}([q]_\sim,w)\in \modif{F'}{F_\sim}$, it holds:
  
  \centerline{
    $\Delta(0,w)\in F$ $\Leftrightarrow$ $\modif{\Delta'}{\Delta_\sim}([0]_\sim,w)\in \modif{F'}{F_\sim}$ and $L(S)=L(\modif{S'}{S_{\sim}})$.
  }
  \end{enumerate}
  
\end{proof}

\subsection{Computing the equivalence relation $\sim$}

In this section, we show how to compute step by step the equivalence relation $\sim$  which merges states of  a QDS.

\begin{definition}\label{def equiv}
  Let $S=(\Sigma,m,\Gamma, {\cal Q}=\bigcup_{k\in [1,m]} Q_k, 0,F,\delta,\gamma)$ be a QDS. Let $p\in Q_i$, $q\in Q_j$. We define  the equivalence relation $\sim _{l,\ l\in \mathbb{N}}$ over $\mathcal{Q}$  by:

    $$\begin{array}{rll}
    \medskip
     p\sim_0 q &\Longleftrightarrow &
    \left\{
      \begin{array}{ll}
\text{(1)}& i=j,\\
\text{(2)}&p\in F \Leftrightarrow q\in F,\\
 \text{(3)}&     \tilde{\gamma}(p)=\tilde{\gamma}(q)\text{ if }i=m
      \end{array}
    \right.\\
  \medskip
     p\sim_{l,\ l>0} q &\Longleftrightarrow &
    \left\{
      \begin{array}{ll}
\text{(1)}&p\sim_{l-1} q,\\
 \text{(2)}&    \left\{ \begin{array}{ll}
 \delta(p,a)\sim_{l-1} \delta(q,a)&\text{ if }p,q\not\in Q_m\\
 \overline{\gamma}(p)\sim_{l-1}\overline{\gamma}(q)&\text{ otherwise}
 \end{array}
 \right.
      \end{array}
    \right.
    \end{array}
    $$

    
\end{definition}

For the purposes of notation, $\sim_l$ designates the set of  class of this equivalence relation. By definition, we have $\sim_{l+1}\subset \sim_l$. It is obvious that $\sim_{l+1}$ is computable from $\sim_l$. We give several lemmas to show that this computation halts and is bounded by $|\cal Q|$. 

\begin{lemma}
Let  $S=(\Sigma,m,\Gamma, {\cal Q}=\bigcup_{k\in [1,m]} Q_k, 0,F,\delta,\gamma)$ be a QDS.  For any integer $l$, 
 $$ \sim_{l+1}=\sim_l\ \Longrightarrow \ \forall k\in \mathbb{N} \sim_{l+k}=\sim_l$$
   \end{lemma}
   \begin{proof}
Let $l$ be the integer such that $\sim_l = \sim_{l+1} \neq\sim_{l+2}$.
Then, there exists $p\in Q_i$,  $q\in Q_j$ such that $p\sim_l q$, $p \sim_{l+1}q $, and $p \not \sim_{l+2} q$.
But $p \not \sim_{l+2} q$ implies that there exists $a \in \Sigma$ such that condition $(2)$ is false (impossible for condition $(1)$).
In this case,
if $p,q\notin Q_m$, $\delta(p,a)\not\sim_{l+1} \delta(q,a)$. But by hypothesis $p\sim_{l+1}q$, then
$\delta(p,a)\sim_{l} \delta(q,a)$ and then $\delta(p,a)\sim_{l+1} \delta(q,a)$ which leads to a contradiction.

If $p,q\in Q_m$, $\overline{\gamma(p)} \not \sim_{l+1} \overline{\gamma(q)}$.
By hypothesis, we have $p \sim_{l+1}q $, that is $\overline{\gamma(p)} \sim_{l} \overline{\gamma(q)}$, and
$\sim_l = \sim_{l+1}$. This would imply $\overline{\gamma(p)} \sim_{l+1} \overline{\gamma(q)}$. Contradiction.
   \end{proof}

%
%
%
%
%
%
\begin{lemma}
Let  $S=(\Sigma,m,\Gamma, {\cal Q}=\bigcup_{k\in [1,m]} Q_k, 0,F,\delta,\gamma)$ be a QDS.   
 $$ \sim_{|{\cal Q}|}=\sim$$
   \end{lemma}
   \begin{proof}

 The smallest integer $l$ such that $\sim_{l}=\sim_{l+1}$ is the highest integer such that $\sim_{l}\subsetneq \sim_{l-1}$. Since the smallest difference between two consecutive equivalences is based on the elimination of only one state, it holds that such an integer $l$ exists and $l\leq |{\cal Q}|$.
  
\end{proof}

\section{QDS can be exponentially smaller than DFAs}\label{sec qds vs dfa}

%

Let $(L_k)_{k\in\mathbb{N}}$ be the family of regular languages on an alphabet $\Sigma$   defined by $L_k=\Sigma^*\cdot\{\sigma\}\cdot\Sigma^k$ with $\sigma\in\Sigma$. This family is known to have a $k+2$-states NFA and a 
$2^{k+1}$
-states minimal DFA.
In this section, we illustrate the factorization power of QDS. We show how to compute a polynomial-size QDS $S_k$  recognizing $L_k$ (See  Figure~\ref{fig ex qds reck}).


\begin{definition}
  Let $k$ be an integer. We denote by $S_k$ the quasi-deterministic structure $(\Sigma,m=k+3, \Gamma, {\cal Q}=\bigcup_{j\in [1,m]} (Q_j\cup Q^{\bullet }_j), 1_1,F,\delta,\gamma )$ defined by:
  \begin{itemize}
    \item $Q_1=\{1_1\}$, 
        \item $\forall j\in\{2,m-1\}$, $Q_j=\{1_j,\ldots , (m-2)_j\}$, $Q^{\bullet}_j= \{1^{\bullet}_j,\ldots , (m-2)^{\bullet}_j\}$,
\item $Q_m=\{1_m,\ldots , (m-1)_m\}$,
    \item $F=\{1_{m-1},\ldots , (m-2)_{m-1}\} \cup \{1_m\}$,
    \item 
    $\delta=
      \left(
        \begin{array}{lll}
         & \bigcup_{a\in\Sigma\setminus \sigma}\{(1_1,a,1^{\bullet}_2)\} \cup \{(1_1,\sigma,1_2)\}&\mathbf{(1)}\\
         \cup & \bigcup_{j\in [1,m-3]}\{(j_{j+1},\sigma,j_{j+2}), (j^{\bullet}_{j+1},\sigma,j^{\bullet}_{j+2})\}&\mathbf{(2)}\\ 
               \cup & \bigcup_{j\in [1,m-3]}\{(j_{j+1},a,(j+1)_{j+2}), (j^{\bullet}_{j+1},a,(j+1)^{\bullet}_{j+2})\mid a\in\Sigma\setminus \{\sigma\}\}&\mathbf{(3)}\\ 
       \cup & \bigcup_{j\in [3,m-2]}\{(i_{j},a,i_{j+1}), (i^{\bullet}_{j},a,i^{\bullet}_{j+1})\mid i\in [1,j-2], a\in \Sigma\}&\mathbf{(4)}\\ 
        \cup & \bigcup_{j\in [1,m-3]}\{(j_{m-1},a,j_{m}), (j^{\bullet}_{m-1},a,j_{m})\mid  a\in \Sigma\}&\mathbf{(5)}\\ 
            \cup & \bigcup_{a\in\Sigma\setminus \sigma}\{((m-2)_{m-1},a,(m-1)_m),((m-2)^{\bullet}_{m-1},a,(m-1)_m)\}&\mathbf{(6)}\\
            \cup & \{((m-2)_{m-1},\sigma,(m-2)_m),((m-2)^{\bullet}_{m-1},\sigma,(m-2)_m)\}&\mathbf{(7)} 
          %
        \end{array}
      \right.$
    \item $\gamma=\bigcup_{j\in[1,m-1]}\{(j_m,j,1_1)\}$
  \end{itemize}
\end{definition}

\begin{example}

The QDS $S_2=(\Sigma,m=5, \Gamma, {\cal Q}=\bigcup_{j\in [1,5]} (Q_j\cup Q^{\bullet }_j), 1_1,F,\delta,\gamma )$ is represented by Figure \ref{fig ex qds reck} where the arrows of $\delta$ are respectively coded by 
$\textcolor{green}{\Arrow[loosely dashdotted]}$ for \textbf{(1)}, $\textcolor{red}{\Arrow[red,decorate,decoration={snake,amplitude=.4mm,segment length=2mm,post length=1mm},densely dotted]}$ for \textbf{(2)}, $\textcolor{red}{\DashedArrow[decorate,decoration={snake,amplitude=.4mm,segment length=2mm,post length=1mm},densely dashed]}$ for \textbf{(3)}, $\Arrow$ for \textbf{(4)}, $\Arrow[decorate,decoration={snake,amplitude=.4mm,segment length=2mm,post length=1mm}]$ for \textbf{(5)}, $\DashedArrow[densely dashdotted]$ for \textbf{(6)} and $\textcolor{green}{\DashedArrow}$ for \textbf{(7)}.
The arrows of $\gamma$ are represented by $\DashedArrow[dotted]$.
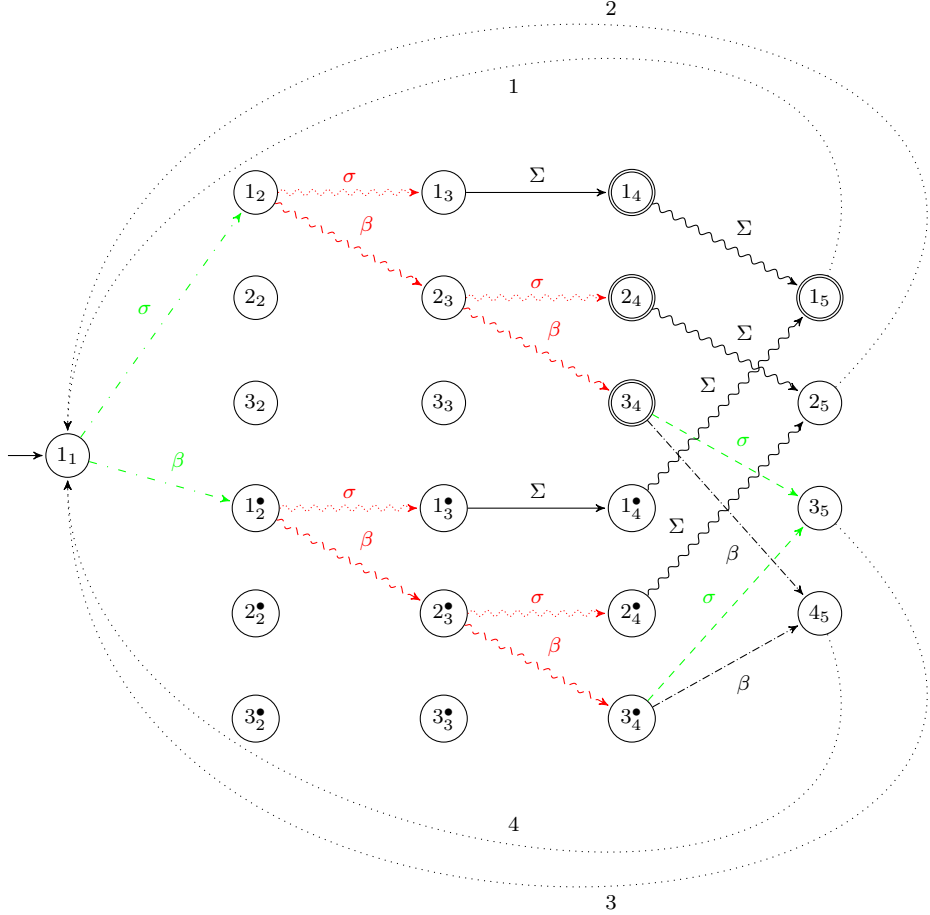
\begin{figure}[H]
    \begin{tikzpicture}[node distance=2cm,bend angle=30,transform shape]   
	    \node[initial, state] (11) at (0,0) {$1_1$};     
	    \node[state] (32) at (2.5,0.7) {$3_2$};     
              \node[state] (22) at (2.5,2.1) {$2_2$};     
              \node[state] (12) at (2.5,3.5) {$1_2$};     
	    \node[state] (12p) at (2.5,-0.7) {$1^\bullet_2$};     
              \node[state] (22p) at (2.5,-2.1) {$2^\bullet_2$};     
              \node[state] (32p) at (2.5,-3.5) {$3^\bullet_2$};     
	    \node[state] (33) at (5,0.7) {$3_3$};     
              \node[state] (23) at (5,2.1) {$2_3$};     
              \node[state] (13) at (5,3.5) {$1_3$};     
	    \node[state] (13p) at (5,-0.7) {$1^\bullet_3$};     
              \node[state] (23p) at (5,-2.1) {$2^\bullet_3$};     
              \node[state] (33p) at (5,-3.5) {$3^\bullet_3$};     
   	    \node[state,accepting] (34) at (7.5,0.7) {$3_4$};     
              \node[state,accepting] (24) at (7.5,2.1) {$2_4$};     
              \node[state,accepting] (14) at (7.5,3.5) {$1_4$};     
	    \node[state] (14p) at (7.5,-0.7) {$1^\bullet_4$};     
              \node[state] (24p) at (7.5,-2.1) {$2^\bullet_4$};     
              \node[state] (34p) at (7.5,-3.5) {$3^\bullet_4$};     
   	    \node[state,accepting] (15) at (10,2.1) {$1_5$};     
              \node[state] (25) at (10,0.7) {$2_5$};     
              \node[state] (35) at (10,-0.7) {$3_5$};     
              \node[state] (45) at (10,-2.1) {$4_5$};     
	    \path[->]
	      (11)   edge[green,loosely dashdotted]  node {$\sigma$} (12)
	                edge[green,loosely dashdotted ] node {$\beta$} (12p)
	        (12) edge[red,decorate,decoration={snake,amplitude=.4mm,segment length=2mm,post length=1mm},densely dotted]  node {$\sigma$} (13)        
	        (12p) edge[red,decorate,decoration={snake,amplitude=.4mm,segment length=2mm,post length=1mm},densely dotted]  node {$\sigma$} (13p)        
  	        (23) edge[red,decorate,decoration={snake,amplitude=.4mm,segment length=2mm,post length=1mm},densely dotted]  node {$\sigma$} (24)        
	        (23p) edge[red,decorate,decoration={snake,amplitude=.4mm,segment length=2mm,post length=1mm},densely dotted]  node {$\sigma$} (24p)        
	        (13) edge  node {$\Sigma$} (14)        
  	        (13p) edge  node {$\Sigma$} (14p)       
	        (12) edge[red,decorate,decoration={snake,amplitude=.4mm,segment length=2mm,post length=1mm},densely dashed]  node {$\beta$} (23)        
	        (23) edge[red,decorate,decoration={snake,amplitude=.4mm,segment length=2mm,post length=1mm},densely dashed]  node {$\beta$} (34)        
	        (12p) edge[red,decorate,decoration={snake,amplitude=.4mm,segment length=2mm,post length=1mm},densely dashed]  node {$\beta$} (23p)        
	        (23p) edge[red,decorate,decoration={snake,amplitude=.4mm,segment length=2mm,post length=1mm},densely dashed] node {$\beta$} (34p)       
	        (14) edge[decorate,decoration={snake,amplitude=.4mm,segment length=2mm,post length=1mm}]  node {$\Sigma$} (15)        
	         (24) edge[decorate,decoration={snake,amplitude=.4mm,segment length=2mm,post length=1mm}]   node {$\Sigma$} (25)        
	        (34) edge[green,dashed]  node {$\sigma$} (35)        
	         (34) edge[densely dashdotted]  node [swap,pos=0.65]{$\beta$} (45)        
	        (34p) edge[green, dashed]  node {$\sigma$} (35)        
	         (34p) edge[densely dashdotted]  node[swap] {$\beta$} (45)    
	         (14p) edge[decorate,decoration={snake,amplitude=.4mm,segment length=2mm,post length=1mm}] node {$\Sigma$}(15)	          
	         (24p) edge[decorate,decoration={snake,amplitude=.4mm,segment length=2mm,post length=1mm}] node[pos=0.3] {$\Sigma$}(25)	          
  	      (15)   edge[dotted,in=90, out=70,looseness=1.3,swap]  node[swap] {$1$} (11)
	      (25)   edge[dotted,in=90, out=45,looseness=2,swap]  node {$2$} (11)
	      (35)   edge[dotted,in=-90, out=-45,looseness=2]  node {$3$} (11)
	      (45)   edge[in=-90, out=-70,dotted,looseness=1.3]  node[swap] {$4$} (11);
 \end{tikzpicture}  
	  \caption{The QDS $S_2$.}
	  \label{fig ex qds reck}
	\end{figure} 
\end{example}

\begin{proposition}\label{prop Sk bon lang}
   For any integer $k$, the QDS $S_k$ recognizes $L_k$.
\end{proposition}
\begin{proof}
  Let $w$ be a word in $\Sigma^n$. Let us denote $\beta=\Sigma\setminus \{\sigma\}$. We show by recurrence on $n$ that  $w\in L(S_k)$ $\Leftrightarrow$ $w\in L_k$.
  
  \textbf{(I)} 
  \begin{itemize} 
  \item Let $w\in \Sigma ^{<k+1}$. By definition, $S_k$ has $k+3$ levels. Final states of $S_k$ are on levels $k+2$ and $k+3$. Thus, $w$ is not recognized by $S_k$. As $L_k=\Sigma^*\{\sigma\}\Sigma^k$, the minimal length of any recognized word is $k+1$. Thus, when $|w|<k+1$, $w\in L(S_k)\Leftrightarrow w\in L_k$. 
 \item Let $w\in \Sigma^{k+1}$. If $w\in L_k$, necessarily $w\in \{\sigma\}\Sigma^k$. By construction, $\delta(1_1,w)\in Q_{k+2}=Q_{m-1}\subset F$. Thus $w\in L(S_k)$. Conversely, if $w\in L(S_k)$, $\delta(1_1,w)\in Q_{m-1}\subset F$. Necessarily, $w[1,1]=\sigma$ and then $w\in L_k$. 
\item Let $w\in \Sigma^{k+2}$.  If $w\in L_k$, necessarily $w\in \Sigma\{\sigma\}\Sigma^k$. By $\mathbf{(1)}$, $\delta(1_1,w[1,1])\in \{1_2,1^{\bullet}_2\}$  and by $\mathbf{(2)}$, as $w[2,2]=\sigma$, $\delta(1_2,w[2,2])=\{1_3\}$ and $\delta(1^{\bullet}_2,w[2,2])=\{1^{\bullet}_3\}$. Moreover, by $\mathbf{(4)}$, $\delta(1_j,w[j-1,j])=1_{j+1}$ and $\delta(1^{\bullet}_j,w[j-1,j])=1^{\bullet}_{j+1}$ for $j\in [3,m-2]$. To conclude, by $\mathbf{(5)}$, $\delta(1_{m-1},w[k+2,k+2])=1_m\in F$ and $\delta(1^{\bullet}_{m-1},w[k+2,k+2])=1_m\in F$, so $w\in L(S_k)$. Conversely, if $w\in L(S_k)$, $\delta(1_1,w)=1_m$. By $\mathbf{(4)}$ there exists $p\in \{1_3, 1^{\bullet}_3\}$ such that $\delta(p,w[3,k+2])=1_m$. By $\mathbf{(2)}$, there exists $q\in \{1_2, 1^{\bullet}_2\}$ such that $\delta(q,w[2,2])=p$ which implies $w[2,2]=\sigma$. As $\{1_2, 1^{\bullet}_2\}$ is directly accessible from the initial state, $w\in L_k$.
  \end{itemize}
  \textbf{(II)} Let us suppose that $n> k+2$. We show by recurrence on $n$ that: 
  \begin{itemize}
    \item $\Delta(1_1,w)=1_2$ $\Rightarrow$ $w[n-k,n]\in \beta^{k}\sigma$\hfill (a)
    \item $\Delta(1_1,w)=1^{\bullet}_2$ $\Rightarrow$ $w[n-k,n]\in \beta^{k+1}$\hfill (b)
    \item $\forall 3\leq j\leq k+2$, $\forall i\leq j-2$, 
      \begin{itemize}
        \item $\Delta(1_1,w)=i_j$ $\Rightarrow$ $w[n-k,n-j+i+2]\in \beta^{k-j+2}\sigma\beta^{i-1}\sigma$ \hfill (c)
        \item $\Delta(1_1,w)=(j-1)_j$ $\Rightarrow$ $w[n-k,n]\in\beta^{k-j+2}\sigma \beta^{j-2}$\hfill (d)
        \item $\Delta(1_1,w)=i^{\bullet}_j$ $\Rightarrow$ $w[n-k,n-j+i+2]\in \beta^{k-j+i+2}\sigma$\hfill (e)
        \item $\Delta(1_1,w)=(j-1)^{\bullet}_j$ $\Rightarrow$ $w[n-k,n]\in \beta^{k+1}$\hfill (f)
      \end{itemize}
    \item $\forall i\leq k+1$, $\Delta(1_1,w)=i_m$ $\Rightarrow$ $w[n-k,n-k+i-1]\in\beta^{i-1}\sigma$\hfill (g)  
    \item $\Delta(1_1,w)=(m-1)_m$ $\Rightarrow$ $w[n-k,n]\in \beta^{k+1}$\hfill (h)
  \end{itemize}
  
  
  Let us suppose that the recurrence  holds for any integer $1\leq n'<n$. Let $w=w'a$ with $a$ in $\Sigma$ and $|w'|=n-1$. Let us suppose that:
  
  \begin{itemize}
    \item[\textbf{(a)}] $\Delta(1_1,w')= 1_2$.  Then  $w'[n-k-1,n-1]\in \beta^{k}\sigma$. By definition, $\Delta(1_1,w)=\delta(\Delta(1_1,w'),a)$.
    If $a=\sigma$, $\Delta(1_1,w)=1_3$ and $w[n-k,n]\in \beta^{k-1}\sigma\sigma$ (case (c) with $i=1$ and $j=3$).
    If $a\in \beta$, $\Delta(1_1,w)=2_3$ and $w[n-k,n]\in \beta^{k-1}\sigma \beta$ (case (d) with $j=3$).
    %
    \item[\textbf{(b)}] $\Delta(i,w')=1^\bullet_2$. By \modif{}{the} recurrence hypothesis, $w'[n-k-1,n-1]\in \beta^{k+1}$. 
    If $a=\sigma$, $\Delta(1_1,w)=1^\bullet _3$ and $w[n-k,n]\in \beta^{k}\sigma$ (case (e) with $i=1$ and $j=3$).
    If $a\in \beta$, $\Delta(1_1,w)=2^\bullet _3$ and $w[n-k,n]\in \beta^{k+1}$ (case (f) with $j=3$).
    %
    \item[\textbf{(c)}] $\Delta(1_1,w')=i_j$ with $3\leq j\leq k+2$ and $i\leq j-2$. If $j\neq k+2$, then by \modif{}{the} recurrence hypothesis, $w'[n-k-1,n-j+i+1]\in \beta^{k-j+2}\sigma \beta^{i-1}\sigma$.  By definition of $\delta$ (4), 
     $\Delta(1_1,w)=i_{j+1}$  and $w[n-k,n-j+i+2]\in \beta^{k-j+1}\sigma \beta^{i-1}\sigma a$. 
    By substituting $j'=j+1$, $\Delta(1_1,w)=i_{j'}$ and $w[n-k,n-j'+i+2]\in \beta^{k-j'+2}\sigma \beta^{i-1}\sigma $ (case (c)).
    If $j=k+2$, then by \modif{}{the} recurrence hypothesis, $w'[n-k-1,n-k+i-1]\in \sigma \beta^{i-1}\sigma$. Hence $\Delta(1_1,w)=i_m$ and $w[n-k,n-k+i-1]\in \beta ^{i-1}\sigma$ (case (g)).
    %
    \item[\textbf{(d)}] $\Delta(1_1,w')=(j-1)_j$ with $3\leq j\leq k+2$.  Then by \modif{}{the} recurrence hypothesis, $w'[n-k-1,n-1]\in \beta^{k-j+2}\sigma \beta^{j-2}$. 
    If $a=\sigma$, $\Delta(1_1,w)=(j-1)_{j+1}$ and $w[n-k,n]\in \beta^{k-j+1}\sigma \beta^{j-2}\sigma$. By substituting $j'=j-1$ and $j''=j+1$, $\Delta(i,w)=j'_{j''}$ and $w[n-k,n]=b^{k-j''+2}\sigma \beta^{j'-1}\sigma$ (case (c)).
    If $a\in \beta$, $\Delta(1_1,w)=j_{j+1}$ and $w[n-k,n]\in \beta^{k-j+1}\sigma \beta^{j-1}$. Then, $\Delta(1_1,w)=(j''-1)_{j''}$ and $w[n-k,n]\in \beta^{k-j''+2}\sigma \beta^{j''-2}$ (case (d)).   
    \item[\textbf{(e)}]  $\Delta(1_1,w')=i^\bullet_j$ with $3\leq j\leq k+2$ and $i\leq j-2$.  If $j\neq k+2$, then by \modif{}{the} recurrence hypothesis, $w'[n-k-1,n-j+i+1]\in \beta^{k-j+i+2}\sigma$. 
     By definition of $\delta$, $\Delta(1_1,w)=i^\bullet_{j+1}$ and $w[n-k,n]\in \beta^{k-j+i+1}\sigma a$. By substituting $j'=j+1$, $\Delta(1_1,w)=i^\bullet_{j'}$ and $w[n-k,n-j'+i+2]\in \beta^{k-j'+i+2}\sigma$ (case (e)).
    If $j=k+2$, then by \modif{}{the} recurrence hypothesis, $w'[n-k-1,n-k+i-1]\in \beta^{i}\sigma$. Hence $\Delta(1_1,w)=i_m$ and $w[n-k,n-k+i-1]\in \beta^{i-1}\sigma$ (case (g)).
     %
    \item[\textbf{(f)}]  $\Delta(1_1,w')=(j-1)^\bullet_j$ with $3\leq j\leq k+2$.  Then by \modif{}{the} recurrence hypothesis, $w'[n-k-1,n-1]\in \beta^{k+1}$. 
    If $a=\sigma$, $\Delta(1_1,w)=(j-1)^\bullet_{j+1}$ and $w[n-k,n]\in \beta^{k}\sigma$. By substituting $j'=j-1$ and $j''=j+1$, $\Delta(1_1,w)=j'^\bullet_{j''}$ and $w[n-k,n]\in \beta^{k-j''+j'+2}\sigma$ (case (e)).
    If $a\in \beta$, $\Delta(1_1,w)=j^\bullet_{j+1}$ and $w[n-k,n]\in \beta^{k+1}$ (case (f)).
    %
    %
    \item[\textbf{(g)}]  $\Delta(1_1,w')=i_m$ with $i\leq m-1$. 
    Then by \modif{}{the} recurrence hypothesis, $w'[n-k-1,n-k+i-2]\in \beta^{i-1}\sigma$. 
    Consequently, $w'[n-k-2,n-1] \in \Sigma \beta^{i-1} \sigma \Sigma^{k+1-i}$, $w[n-2-k,n] \in \Sigma \beta^{i-1} \sigma \Sigma^{k+2-i}$ and $w[n-2-k+i,n] \in \sigma \Sigma^{k+2-i}$.
    By definition of $\gamma$, $\gamma(i_m)=(i,1_1)$.
    Furthermore, $\Delta(1_1,w')=\delta(1_1,w'[n-2-k,n-1])$.
    Hence $\Delta(1_1,w)=\delta(1_1,w[n-2-k+i,n])$. So $w[n-k-2+i,n]\in \sigma\Sigma^{k+2-i}$ and the cases (c), (d) and (g) have to be considered.
    %
%
%
    Hence \modif{}{the} recurrence holds. 
    
    \item[\textbf{(h)}] $\Delta(1_1,w)=(m-1)_m$.  Then by \modif{}{the} recurrence hypothesis, $w'[n-k-1,n-1]=\beta^{k+1}$. 
    If $a=\sigma$, $\Delta(1_1,w)=1_2$ and $w[n-k,n]\in \beta^{k}\sigma$ (case (a)). 
    If $a\in \beta$, $\Delta(1_1,w)=1^\bullet_2$ and $w[n-k,n]\in \beta^{k+1}$ (case (b)). 
    Hence \modif{}{the} recurrence holds. 
  \end{itemize}
  
  
  \textbf{(III)} Finally, since $q\in F \Leftrightarrow q=i_j \wedge (j=m-1 \vee (i=1\wedge j=m))$, and since for any $w$, $\Delta(1_1,w)$ is defined, it holds that $w\in L(S_k)$ $\Leftrightarrow$ $w\in L_k$.
  
\end{proof}

\modif{}{As a direct consequence of Proposition~\ref{prop Sk bon lang}, it holds:}
\begin{theorem}\label{thm qds expo smaller}
  There exists a QDS $S_k$ that recognizes $L_k$ which is exponentially smaller than the minimal DFA associated with $L_k$.
%
\end{theorem}
\modif{\begin{proof}
  Corollary of Proposition~\ref{prop Sk bon lang}.
  
\end{proof}}{}

\section{Conclusion and Perspectives}

    Quasi-deterministic structures are an alternative to the computation of a deterministic automaton since they can be used as recognizers, while reducing the space needed to solve the membership problem once computed.
We have  shown that reduced QDSs can be exponentially smaller than equivalent deterministic automata. 


 
 
A regular language is $(k,l)$-unambiguous if it is denoted by some regular expression the position automaton of which is $(k,l)$-unambiguous. 
Similar extensions were already defined for deterministic automata ($1$-unambiguity). 
Denoting a language by such an expression allows us to directly compute a quasi-deterministic structure in order to solve the membership problem.  
  
One may wonder whether every regular language admits a $(k,l)$-unam\-biguous position automaton recognizing it.
If the answer is negative, then a second question arises: Is it possible to characterize languages having a $(k,l)$-unam\-biguous position automaton, as Br\"uggemann-Klein and Wood did for languages having a deterministic position automaton ?


\section*{References}
\bibliography{/Users/pacot/TEX/BIBLIO/biblio}

\end{document}